\documentclass[accepted]{article}

\usepackage{microtype}
\usepackage{graphicx}
\usepackage{subfigure}
\usepackage{booktabs} 

\usepackage{hyperref}

\usepackage[ruled,vlined]{algorithm2e}
\usepackage{subcaption}
\usepackage{wrapfig}
\usepackage{bbding}
\usepackage{bm}
\usepackage{icml2024}

\usepackage{amsmath}
\usepackage{amssymb}
\usepackage{mathtools}
\usepackage{amsthm}

\usepackage[capitalize,noabbrev]{cleveref}

\theoremstyle{plain}
\newtheorem{theorem}{Theorem}[section]

\newtheorem{corollary}[theorem]{Corollary}
\newtheorem{definition}[theorem]{Definition}

\theoremstyle{remark}

\usepackage[textsize=tiny]{todonotes}

\newcommand{\Appendix}[1]{the full version for}

\newcommand{\x}{\bm{x}}

\newcommand{\s}{\bm{s}}

\newcommand{\K}{\mathcal{K}}

\renewcommand{\P}{\mathcal{P}}

\newcommand{\cV}{\mathcal{V}}

\newcommand{\bbE}{\mathbb{E}}

\newcommand{\bbR}{\mathbb{R}}

\icmltitlerunning{A Resilient and Accessible Distribution-Preserving Watermark for Large Language Models}

\begin{document}

\twocolumn[
\icmltitle{A Resilient and Accessible Distribution-Preserving Watermark\\ for Large Language Models}




\begin{icmlauthorlist}
\icmlauthor{Yihan Wu}{yyy}
\icmlauthor{Zhengmian Hu}{yyy}
\icmlauthor{Junfeng Guo}{yyy}
\icmlauthor{Hongyang Zhang}{sch}
\icmlauthor{Heng Huang}{yyy}
\end{icmlauthorlist}

\icmlaffiliation{yyy}{Department of Computer Science, University of Maryland College Park}
\icmlaffiliation{sch}{School of Computer Science, University of Waterloo}

\icmlcorrespondingauthor{Yihan Wu}{ywu42@umd.edu}
\icmlcorrespondingauthor{Heng Huang}{heng@umd.edu}

\icmlkeywords{Machine Learning, ICML}

\vskip 0.3in
]



\printAffiliationsAndNotice{}  

\begin{abstract}
Watermarking techniques offer a promising way to identify machine-generated content via embedding covert information into the contents generated from language models. A challenge in the domain lies in preserving the distribution of original generated content after watermarking. Our research extends and improves upon existing watermarking framework, placing emphasis on the importance of a \textbf{Di}stribution-\textbf{P}reserving (DiP) watermark. Contrary to the current strategies, our proposed DiPmark simultaneously preserves the original token distribution during watermarking (distribution-preserving), is detectable without access to the language model API and prompts (accessible), and is provably robust to moderate changes of tokens (resilient). DiPmark operates by selecting a random set of tokens prior to the generation of a word, then modifying the token distribution through a distribution-preserving reweight function to enhance the probability of these selected tokens during the sampling process. Extensive empirical evaluation on various language models and tasks demonstrates our approach's distribution-preserving property, accessibility, and resilience, making it a effective solution for watermarking tasks that demand impeccable quality preservation. 
Code is available at\footnote{https://github.com/yihwu/DiPmark.git}.
\end{abstract}
\section{Introduction}
In the current era, artificial intelligence has attained the capability to generate text remarkably indistinguishable from human authorship \cite{palm2,openai2023gpt}. This advancement has raised concerns regarding the discernment of authenticity in content, questioning whether it originates from human intellect or AI models. In particular, the proficiency of large language models (LLMs) in imitating human writing style brings a series of implications. While these models facilitate the simplification of complex tasks and enhance human capabilities, they simultaneously harbor risks of misuse, evident in instances of academic dishonesty and the spread of misinformation via online platforms.

The challenge of distinguishing machine-generated content from that authored by humans is escalating, with conventional detection tools often proving inadequate~\cite{krishna2023paraphrasing}. 
To address this issue, \textit{watermarking} emerges as a nuanced solution~\cite{kirchenbauer2023watermark}. This type of approach involves embedding discreet yet identifiable watermarks in AI-generated text, signifying its artificial origin. Beyond the widely held notion that watermarks should be identifiable via a secret key~\cite{kirchenbauer2023watermark}, there are additional fundamental characteristics necessary for an efficient watermark within language models:
\vspace{-0.2cm}
\begin{itemize} 
    \item (Distribution-preserving) The watermark should provably preserving the distribution of the original language model.
    \item (Accessible) Detecting watermark within the content should be efficient and straightforward without accessing the language models and prompts. 
    \item (Resilient) The watermark should remain identifiable if the content undergoes moderate modifications. Furthermore, we define a watermark as `provably resilient' if it can be provably identified under such modifications.

\end{itemize}
\vspace{-0.2cm}


To the best of our knowledge, there is no watermark technique adhere to the aforementioned three key properties simultaneously (see Table~\ref{tab:watermark comparison} for an overall comparison).
Existing methods either impact the model's sampling distribution \citep{kirchenbauer2023watermark,zhao2023provable}, lack resilience against text alterations such as editing or cropping \citep{christ2023undetectable}, require thousands of inference step during the detection process \citep{kuditipudi2023robust}, or require the prompt and the token logits of language model API during detection \citep{hu2023unbiased}.

\begin{table*}[]
\centering
\vspace{-0.2cm}
\caption{Existing watermarking techniques do not adhere to all three key properties (distribution-preserving, accessible, resilient). \textbf{Distribution-preserving:} \citet{kirchenbauer2023watermark} impacts the distribution of the generated tokens. \textbf{Accessible:} During detection, \citet{kuditipudi2023robust} necessitates thousands of inference steps, and \citet{hu2023unbiased} requires the token logits of language model API and the prompt, which could result in huge computational costs and hurt the accessibility. \textbf{Resilient and Provably Resilient:} DiPmark is provably resilient against arbitrary text modifications with a guaranteed false positive rate, whereas other methods lack corresponding discussions.}
\label{tab:watermark comparison}
\resizebox{1\textwidth}{!}{%
\begin{tabular}{l|cccc}
\toprule
Properties
& \citet{kirchenbauer2023watermark}  & \citet{kuditipudi2023robust} &\citet{hu2023unbiased}& DiPmark \\ \midrule
Distribution-preserving (Sec.~\ref{sec:DiPmark}\&~\ref{sec:dip perf exp})                     & \XSolidBrush  & \Checkmark    & \Checkmark             &\Checkmark            \\ 
Accessible (Sec.~\ref{sec:detection}\&~\ref{sec:Efficiency})    &\Checkmark                  & \XSolidBrush & \XSolidBrush          &\Checkmark            \\ 
Resilient and Provably Resilient (Sec.~\ref{sec:provable robust}\&~\ref{sec:robust})               & \XSolidBrush     & \XSolidBrush       & \XSolidBrush           &\Checkmark            \\
 \bottomrule
\end{tabular}
}
\vspace{-0.4cm}
\end{table*}
Our watermarking framework (\textit{i.e.,} DiPmark), in alignment with pre-existing schema~\cite{kirchenbauer2023watermark}, is comprised of two components: (1.) a generating function 
, which transforms a prompt and a secret watermark key into the content from the language model; and (2.) a detecting function that identifies a potential watermarked text through the secret key. During the text generation process, language model providers will adjust the output probability of the generated tokens using a secret key. We design a novel distribution-preserving generating function, ensuring that each instance of text generation consists with the original language model's distribution. As for the detection phase, the user can detect the presence of watermark efficiently by solely using the secret key and the watermarked text without accessing prompts and language model API. Through experimental assessments on widely-studied language models, including BART-large model~\citep{liu2020multilingual}, LLaMA-2~\citep{touvron2023llama2}, and GPT-4~\cite{openai2023gpt}; our approach is demonstrated possessing above mentioned three fundamental properties. 



\textbf{Our contributions.}
Our work tackles the problem of designing watermarks for large language models without affecting its overall performance and advances the state-of-the-art in multiple ways.

\begin{itemize}
\vspace{-0.2cm}
\item We propose a novel watermarking framework, DiPmark, that introduces a \textbf{provably} distribution-preserving watermarking scheme for language models. Comparing with existing methods, DiPmark is \textbf{simultaneously} distribution-preserving, efficient, and provable resilient.

\item We identify the existing watermark detector \cite{kirchenbauer2023watermark} cannot precisely guarantee the false positive rate of detection. To solve this problem, we develop an well-defined watermark detection statistic for DiPmark, which can reliably detect the watermark within generated contents while maintaining a guaranteed false positive rate. 
Furthermore, we also show our detect algorithm is \textbf{provably} robust against arbitrary text modifications.

\item Through extensive experiments on widely-adopted language models 
, we validate the distribution-preserving property of DiPmark.
Notably, the detection time for 1,000 watermarked sequences produced by LLaMA-2 stands at a mere 90 seconds without the need of API access and prompts (at least 4X faster compared with current distribution-preserving watermark detection~\citep{hu2023unbiased,kuditipudi2023robust}). Furthermore, DiPmark exhibits robustness even when subjected to 20\% to 30\% random text modifications and paraphrasing attacks. Finally, in a case study, we show the effectiveness of DiPmark on GPT-4.

\vspace{-0.2cm}

\end{itemize}

\section{Related Work}
In a recent seminal work, \citet{kirchenbauer2023watermark} introduced a pioneering watermarking scheme tailored for LLMs. 
However, this approach inevitably leads to a pivotal change in the distribution of the generated text, potentially compromising the quality of the generated content. To maintain the output distribution in watermarked content, alternative strategies have been explored. \citet{christ2023undetectable} and \citet{kuditipudi2023robust} employed the inverse sampling method to generate watermarked token distributions. Notably, \citet{christ2023undetectable}'s method faces resilience issues under modifications or changes and lacks empirical validation for detectability. Meanwhile, \citet{kuditipudi2023robust}'s approach requires the secret key distribution during detection, potentially compromising data security and watermark stealthiness. Moreover, their detection process involves thousands of resampling steps from the secret key distribution, which is inefficient for lengthy texts. \citet{hu2023unbiased} also used inverse sampling and permutation based reweight for watermarking, but the detector requires the token logits of language model API and the prompt for generating the content, undermining its operational efficiency. A detailed discussion of watermarking LLMs is in Appendix~\ref{app:related work}. 

Our research aligns closely with \citet{kirchenbauer2023watermark}. In their settings, they employed watermarking for text derived from a language model by separating the token set into `red' and `green' lists. Building on this foundation, we introduce an evolved family of reweight strategies. This approach ensures equivalency in distribution between the watermarked language model and the original language model. 
\section{Preliminary} 

\textbf{Notations.} 
We first introduce a few essential notations. Let us represent the vocabulary (or token) set by $V$ and its size or volume by $N=|V|$. We further introduce the set $\cV$, defined as an aggregation of all string sequences, even accounting for those of zero length.
In the context of a language model, it produces a token sequence based on a given prompt. For a single step of this process, the likelihood of generating the next token 
$x_{n+1}\in V$ conditioned on the current context $x_1,...,x_n$ is represented as 
$P_M(x_{n+1}\mid x_1, x_2, ..., x_n)$. 
For the sake of brevity and clarity, we opt for the condensed notation: 
$P_{M}(\bm{x}_{n+1:n+m}\mid\bm{x}_{1:n})$, where $\bm{x}_{n+1:n+m}=(x_{n+1},\dots,x_{n+m})$. Note that the prompt is deliberately omitted in this representation.

In the context of watermarking, the server provider will use a set of \textit{i.i.d.} \textit{watermark cipher} $\{\theta_i\in\Theta,i\in\mathbb{N}\}$ on the cipher space $\Theta$ to generate the text. The cipher $\theta_i$ is usually generated by a secret key $k\in\K$ and a \textit{fragment} of the previous context, named \textit{texture key}, $\s_i$.  Instances of texture keys include $x_{t-1}$, $\x_{t-3:t-1}$, $\x_{1:t-1}$, etc. Each $\theta_i$ is independent and following the same distribution $P_{\Theta}$. We now provide the formal definition of the reweight strategy.

\begin{definition}[Reweight strategy] Denote by $\P$ the set of all distributions on the token set $V$. A reweight strategy is a mapping $P_W:\P\times\Theta\to\P$. Given the original distribution $P_M(x_{n+1}\mid \x_{1:n})\in\P$, the watermarked distribution with cipher $\theta_i$ is given by $P_W(P_M(x_{n+1}\mid \x_{1:n}),\theta_i)$. For brevity, we represent it as $P_W(x_{n+1} | \x_{1:n}, \theta_i)$.
\end{definition}

The reweight strategy stands as the foundation of the watermark algorithm by shaping the distribution of watermarked text. As introduced in \cite{kirchenbauer2023watermark}, the authors propose a red-green list reweight technique, where the vocabulary set is separated into the red and green lists and the probability of green tokens is promoted during the sampling process. Specifically, given an initial token probability $p(t)$, the watermarked probability for the token, denoted by $p_W(t)$, is formulated as:
\begin{equation*}p_W(t)=\left\{
    \begin{aligned}
        &\frac{p(t)}{\sum_{t\in\textrm{red}}p(t)+\sum_{t\in\textrm{green}}e^{\delta}p(t)},\quad t\in \textrm{red list};\\
        &\frac{e^{\delta}p(t)}{\sum_{t\in\textrm{red}}p(t)+\sum_{t\in\textrm{green}}e^{\delta}p(t)},\quad t\in \textrm{green list},
    \end{aligned}\right.
\end{equation*}
where $\delta > 0$ is a predetermined constant. This strategy reveals an inherent bias in the watermarked distribution. For example, consider $\gamma = 0.5$, suggesting that half of $V$ comprises the red list. With $V = \{a,b\}$, and given probabilities $p(a) = 0.99$ and $p(b) = 0.01$, there are two equivalent permutations of $V$ with congruent appearance likelihoods. An analysis for any value of $\delta > 0$ yields $p_W(a) = 0.5(\frac{e^{\delta}p(a)}{e^{\delta}p(a) + p(b)} + \frac{p(a)}{e^{\delta}p(b) + p(a)}) < p(a)$. This indicates that the red-green list watermark does not preserve the original text's probability. Below we introduce the formal definition of distribution-preserving reweight strategy and distribution-preserving watermark.


\begin{definition}[Distribution-preserving reweight strategy]
    A reweight strategy, denoted $P_W$, is said to be distribution-preserving at an individual generation step if, for all $\x_{1:n} \in \cV$ and any $i \leq n$, it holds that $P_M(x_i|\x_{1:i-1}) = \bbE_{\theta_i\sim P_\Theta}[P_W(x_i|\x_{1:i-1},\theta_i)].$
    \end{definition}
\begin{definition}[Distribution-preserving watermark]
If a watermark framework preserves the text distribution throughout all generation steps, i.e., $\forall n>0$, for all sequences $\x_{1:n} \in \cV$  we have $P_M(\x_{1:n}) = \bbE_{\theta_1,...,\theta_n}[P_W(\x_{1:n}|\theta_1,...,\theta_n)],$ then the watermark is distribution-preserving. 
\end{definition}

A distribution-preserving reweight strategy can naturally lead to a distribution-preserving watermark, as illustrated by:
\vspace{-0.2cm}
\begin{equation*}
    \begin{split}
        &\bbE_{\theta_{1:n}}[P_W(\x_{1:n}|\theta_{1:n})] = \bbE_{\theta_{1:n}}\left[\prod_{i=1}^nP_W(x_i|\x_{1:i-1},\theta_i)\right] \\ 
        &= \prod_{i=1}^n\bbE_{\theta_i}[P_W(x_i|\x_{1:i-1},\theta_i)] = P_M(\x_{1:n}).
    \end{split}
\end{equation*}

The above equality stems from the independence property of the set $\{\theta_i\}$. Therefore, to establish a distribution-preserving watermark, it is essential to incorporate both: a) a distribution-preserving reweight strategy and b) an \textit{i.i.d.} set of ciphers, $\{\theta_i\}$.

We emphasize the significance of preserving the distribution of text during watermarking, motivated by the following justifications: a) Stealthy Watermarking: A watermark that disrupts the original distribution of a language model lacks the attribute of stealthiness. Such alterations make it relatively straightforward to distinguish between watermarked and unwatermarked LMs through multiple instances of sampling. b) Industry-Level LLM Application: When contemplating the application of a watermark to industry-standard LLMs like ChatGPT and Bard, the primary consideration is to ensure that the watermark does not compromise the performance of these foundational LLMs. Any watermark that interferes with the original text distribution will inevitably impact the quality of generated text, an outcome that is unacceptable by industry stakeholders.

 In the next section, we introduce a reweight strategy with a distribution-preserving characteristic. This attribute guarantees that the text distribution remains unaltered even as we enhance the utilization of tokens from the green list during the watermarking process.

\begin{figure}
\centering
\includegraphics[width=0.5\textwidth]{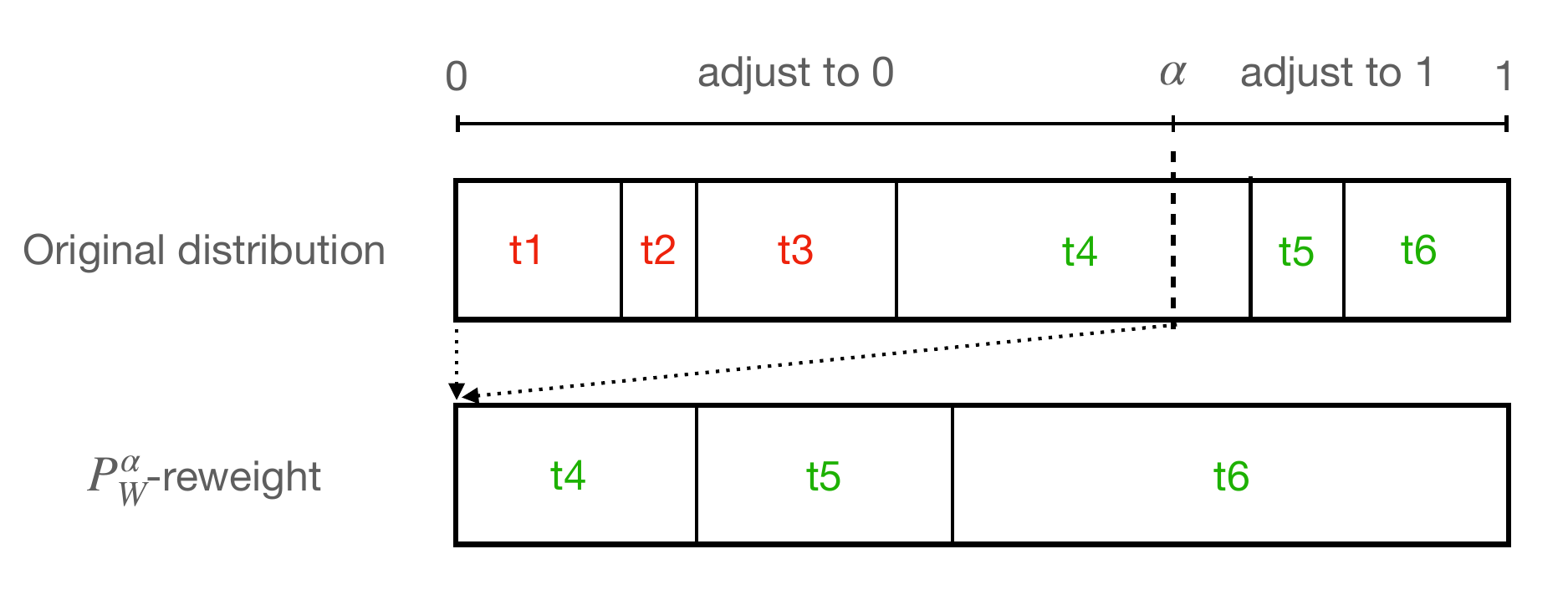}
\vspace{-10pt}
\includegraphics[width=0.5\textwidth]{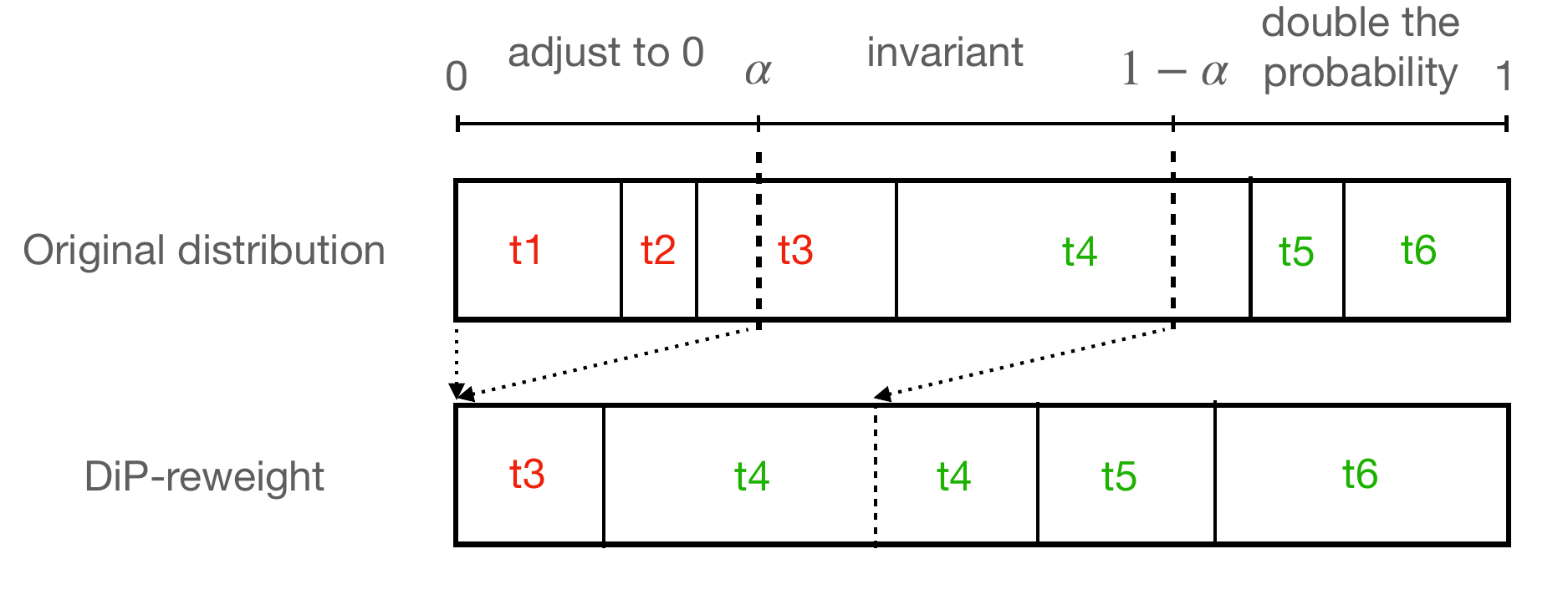}
\vspace{-15pt}
\caption{Illustration of the $P_W^\alpha$-reweight and DiP-reweight. \textbf{Top.} In $P_W^\alpha$-reweight, the token probabilities within the interval $[0, \alpha]$ are adjusted
to 0, while the rest are adjust to 1. \textbf{Bottom.} In DiP-reweight, the probability mass within $[0,\alpha]$ is transferred to the probability mass within $[1-\alpha,1]$.}
\label{fig:violin plot main}
\vspace{-13pt}
\end{figure}
\section{DiPmark}\label{sec:DiPmark}

\textbf{Motivation.} The reweight strategy presented in \citet{kirchenbauer2023watermark} disrupts the inherent text distribution when promoting the use of the green tokens during the sampling process. Such disruption would lead to biased sampling, seriously affecting the quality of the generated text. To address this issue, we design a novel reweight strategy that ensures the token distribution remains unaltered during the watermarking process. Contrary to the approach in \cite{kirchenbauer2023watermark} that promotes the use of all tokens from the green list, we emphasize increasing the \textit{sum of the probability} of the green-list tokens. In this way, the watermarked text, when exposed to the secret key, will still exhibit a bias towards the green-list tokens. Motivated by that, we design a reweight function, which preserves the text distribution during watermarking process. 

\textbf{Cipher space for watermarking.} 
Our considered watermark cipher space encompasses the permutations of the vocabulary set, denoted as $\Theta = \{V^{p}_1,...,V^{p}_{N!}\}$, wherein $V^{p}_i$ represents a permutation of $V$. As for the cipher distribution $P_\Theta$, we employ a uniform distribution over $\Theta$, ensuring that each permutation is equally probable for selection.

\textbf{Reweight strategy.} Let $\theta\in\Theta$ be a cipher, constituting a permutation of $V$. The probabilities of individual tokens can be arranged within the interval $[0,1]$ according to their respective positions in $\theta$. Given a fixed constant $\alpha$ in $[0,1]$, the token probabilities within the interval $[0,\alpha]$ are adjusted to $0$, while those in the interval $[\alpha,1]$ are scaled by a factor of $\frac{1}{1-\alpha}$.
 Let $\gamma\in[0,1]$ be the red-green list separator for the permuted token list, which is in accordance with the definition in \citet{kirchenbauer2023watermark}.
Through this reweight strategy, we can increase the sum of the probability of green-list tokens for \textit{arbitrary} permutation separator $\gamma$, as the green-list tokens consistently appear towards the end of the ordered set $\theta$. Below we present the formal definition of our reweight strategy.
\begin{definition}[$P_W^\alpha$-reweight strategy]
    Let $\theta=\{t_1,...,t_N\}$, which represents a permutation of $V$, and denote $P_M(\cdot|\x)$ as the original token distribution. Let $F^\alpha(i|\theta) := \frac{1}{1-\alpha}\max\{\sum_{j=1}^{i}P_M(t_j|\x)-\alpha,0\}$. The $P_W^\alpha$-reweight probability distribution is $P_W^\alpha(t_i|\x,\theta) = F^\alpha(i|\theta)-F^\alpha(i-1|\theta)$.
\end{definition}
It is easy to show that $P_W^\alpha(t_i|\x,\theta)$ is a distribution on $V$ for arbitrary $\alpha$. Firstly, as $F^\alpha(i|\theta)$ is monotonously increasing with $i$, we have $P_W^\alpha(t_i|\x,\theta) = F^\alpha(i|\theta)-F^\alpha(i-1|\theta)\geq 0$. Secondly, the sum of the probability of all tokens is  $\sum_{i=1}^{N}P_W^\alpha(t_i|\x,\theta) = \sum_{i=1}^{N}(F^\alpha(i|\theta)-F^\alpha(i-1|\theta)) = F^\alpha(N|\theta)=1$.

We wish to highlight the distinction between the probability quantile $\alpha$ and the red-green list separator $\gamma$. $\gamma$ serves as the partition for the permuted token list. In contrast, $\alpha$ separates the \textit{probability interval} $[0,1]$ of the permuted token list. Thus, both the $P_W^\alpha$-reweight and DiP-reweight (as subsequently defined) remain oblivious to $\gamma$, while still effectively promoting the probability of green list tokens.

Leveraging the symmetry of permutations, we can prove that a weighted combination of $P_W^\alpha$-reweight and $P_W^{1-\alpha}$-reweight yields a distribution-preserving reweight strategy. It is pivotal to recognize that both $P_W^\alpha$-reweight and $P_W^{1-\alpha}$-reweight increase the sum of the probability of green-list tokens. Therefore, the combined effect of these reweight functions still exhibits a preference for the green list tokens. The formal definition of our distribution-preserving reweight strategy is presented subsequently.
\begin{algorithm}[t]
\caption{DiPmark generator}\label{alg:DiPmark generator}
\SetAlgoLined
\KwIn{
 watermark key $k$, reweight parameter $\alpha$, prompt $\bm{x}_{-m:0}$, generate length $n\in\mathbb{N}$, context window length $a$, and permutation generation function $h$.}
 
Initialize texture key history $hist$.\\
\For{$i=1,\dots,n$}{
        Calculate the LM distribution for generating the $i$-th token $P_M(\cdot\mid\bm{x}_{-m:i-1})$. \\
        Generate a texture key $\s_{i}$ from $\x_{i-a:i-1}$.\\
    \eIf{$\s_{i} \in hist$}{
        Sample the next token $x_{i}$ using distribution $P_M(\cdot\mid\bm{x}_{-m:i-1})$.}
    {
        Update key history $hist.append(\s_{i})$}
    Generate the cipher $\theta_i=h(k,\s_{i})$.
    Sample the next token $x_{i}$ using distribution $P_W(\cdot|\x_{-m:i-1},h(k,\s_{i}))$.
}
\textbf{return} $\bm{x}_{1:n}$.
\end{algorithm}

\begin{definition}[DiP-reweight strategy]\label{def:dipreweight} Denote by $\theta=\{t_1,...,t_{N}\}$ the cipher, which is a permutation of $V$.
Given the original token distribution $P_M(t|\x),\forall t\in V$, where $\x\in\Sigma$ is the previous token sequence,
the DiP-reweight strategy is represented by $$P_W(t_i|\x,\theta):= (1-\alpha)P_W^\alpha(t_i|\x,\theta)+\alpha P_W^{1-\alpha}(t_i|\x,\theta).$$
\end{definition}
As both $P_W^\alpha$ and $P_W^{1-\alpha}$ are distributions on $V$ and $P_W(t_i|\x,\theta)$ is a convex combination of them, $P_W(t_i|\x,\theta)$ is also a distribution on $V$.

\begin{theorem}\label{thm:dipreweight}
     DiP-reweight is a distribution-preserving reweight strategy, i.e., for all $\x_{1:n} \in \cV$ and any $i \leq n$, it holds that $P_M(x_i|\x_{1:i-1}) = \bbE_{\theta_i\sim P_\Theta}[P_W(x_i|\x_{1:i-1},\theta_i)].$
\end{theorem}
We defer the proof of Theorem \ref{thm:dipreweight} to Appendix~\ref{sec:missing proof}. With the DiP-reweight approach, the generation of \textit{i.i.d.} ciphers, denoted as $\theta_i$, becomes essential for crafting a distribution-preserving watermark. Let $k$ represent a stochastic secret key derived from the key space $K$ following the distribution $P_K$, let $\s\in\cV$ be a texture key, which is a sub-sequence of the previously generated context. 
Denoted by $\x_{1:t-1}$ the context generated prior to time step $t$ , instances of texture keys encompass $x_{t-1}$, $\x_{t-3:t-1}$, and $\x_{1:t-1}$.
We introduce a hash function, $h(k,\s): K\times\cV\to \Theta$, orchestrating the mapping of a secret key in conjunction with a texture key. $\s\in\cV$ to a permutation of the token set $V$. In order to achieve distribution-preserving watermarking, the chosen hash function $h$ should adhere to the following conditions: 
a) For distinct (secret key, texture key) pairs, i.e., $(k_1,\s_1)\neq(k_2,\s_2)$, $h(k_1,\s_1)$ ought to be statistically independent from $h(k_2,\s_2)$, 
and 
b) Upon holding $\s$ constant, every $V^{p}_i\in\Sigma$ should exhibit a uniform likelihood of being selected given a random key, specifically, $\forall V^{p}_i\in\Sigma, \bbE_{k\sim P_K}[\bm{1}_{h(k,\s)=V^{p}_i}]=1/N!$. 

There exists hash functions meeting the above criteria, one example being the hash function introduced in \citet{kirchenbauer2023watermark}. Under such conditions, the cipher $\theta_i$ can be deemed \textit{i.i.d.} if the texture key $\s_i$ is distinctive for each instance. To ensure this uniqueness, a historical log is employed to retain texture keys generated in prior steps. If a texture key is identified in the historical log, another secret key will be utilized with the texture key to generate the cipher. The detailed methodology is shown in Alg.~\ref{alg:DiPmark generator}.

\begin{corollary}\label{col:disprev}
    DiPmark (Alg. \ref{alg:DiPmark generator}) is a distribution-preserving watermark, i.e., for all sequences $\x_{1:n} \in \cV$ and any positive integer $n$, we have $P_M(\x_{1:n}) = \bbE_{\theta_1,...,\theta_n}[P_W(\x_{1:n}|\theta_1,...,\theta_n)].$
\end{corollary}
This can be easily validated by combining the distribution-preserving property of DiP-reweight and the independence of ciphers $\theta_i$.

\begin{algorithm}[t]
\caption{DiPmark detector}\label{alg:DiPmark detector}
\SetAlgoLined
\KwIn{
text $\x_{1:n}$, watermark key $k$, volume of the token set $N$, permutation generation function $h$, green list separator $\gamma$, context window length $a$, and threshold $z$.}
    Initialize the green token indexer of $\gamma$: $L_G(\gamma)=0$.\\
    \For{$i = 2,...,n$}{
    Generate a texture key $\s_{i}$ based on $\x_{i-a:i-1}$.\\
    Generate the permutation of token set $\theta_i=h(k,\s_{i})$.\\
    Calculate the list of green tokens via $G = \theta_i[\lceil \gamma N\rceil:N]$.\\
    \textbf{if} $x_{i} \in G$:\\
    \quad\quad $L_G(\gamma)=L_G(\gamma)+1$.
    }
     Calculate the score:
    $\Phi(\gamma,\x_{1:n}) = \frac{L_G(\gamma)}{n}-(1-\gamma)$.

\textbf{return} $\Phi(\gamma,\x_{1:n})>z$.
\end{algorithm}
\vspace{-0.2cm}
\section{DiPmark Detection}\label{sec:detection}
\vspace{-0.1cm}
We leverage a hypothesis test to identify the presence of DiPmark. In the context of a predetermined red-green list separator $\gamma\in [0,1]$, we classify the initial $\lceil\gamma N\rceil$ tokens within the token set permutation as belonging to the red list, while the remaining tokens are categorized as part of the green list. Given a text sequence $\x_{1:n}$, we establish the null hypothesis $H_0$: \textit{$\x_{1:n}$ is generated without any awareness of DiPmark.} Below we design a statistic, named ``green token ratio'', for conducting the hypothesis test.

\begin{definition}[Green token ratio]\label{def:green token ratio} Let $L_G(\gamma)$ be the count of green tokens within $\x_{1:n}$, where $\gamma$ is the predetermined red-green list separator. The green token ratio is give by
    $\Phi(\gamma,\x_{1:n}):=L_G(\gamma)/n-(1-\gamma).$
\end{definition}

The green token ratio quantifies the bias towards green tokens within the text sequence. The term $L_G(\gamma)/n$ signifies the proportion of green tokens within a sequence of tokens, while $1-\gamma$ denotes the expected green token proportion in an unwatermarked sequence. Under the null hypothesis $H_0$, $L_G(\gamma)$ follows a binomial distribution with parameters $p=(1-\gamma)$ and $n$ total trials, i.e., $L_G(\gamma)\sim\textrm{Binomial}(n,1-\gamma)$. 
The reason for this is that each token is randomly assigned to either the red or green list in the absence of our watermarking rule. We derive the subsequent concentration bound of the green token ratio $\Phi(\gamma,\x_{1:n})$:
\begin{theorem}
    [Concentration bound of $\Phi(\gamma,\x_{1:n})$]\label{thm:concenbound}
    Let $\Phi(\gamma,\x_{1:n}):=L_G(\gamma)/n-(1-\gamma)$, where $L_G(\gamma)\sim\textrm{Binomial}(n,1-\gamma)$. We have $\forall t\in\bbR$, $$\Pr(\Phi(\gamma,\x_{1:n})\geq t)\leq \exp(-n\mathbb{KL}(t+1-\gamma||1-\gamma)),$$
where $\mathbb{KL}(p||q):= p\log\frac{p}{q}+(1-p)\log\frac{1-p}{1-q}$ is the Kullback-Leibler divergence.
\end{theorem}
We proceed to reject the null hypothesis and detect the watermark if $\Phi(\gamma,\x_{1:n})$ surpasses a predefined threshold. For instance, setting the threshold as $\Phi(\gamma,\x_{1:n})\geq1.517/\sqrt{n}$ results in rejecting $H_0$ (indicating watermark presence) while maintaining a false positive rate below 1\%. Our detection algorithm is shown in Alg.~\ref{alg:DiPmark detector}.
Noting that the concentration bound of $\Phi(\gamma,\x_{1:n})$ scales proportionally with $n$ times the green token ratio. With a fixed green token ratio $\Phi(\gamma,\x_{1:n})$, detecting longer sequences becomes more straightforward because they will show a lower false positive rate. The validity of this analysis is also confirmed in Section~\ref{sec:add2}.

\textbf{Difference between our detection algorithm and \citet{kirchenbauer2023watermark}.} It is noteworthy that we diverge from \citet{kirchenbauer2023watermark} by avoiding the use of the z-test statistic $(L_G(\gamma)-(1-\gamma)n)/\sqrt{n\gamma(1-\gamma)}$. The z-test assumes a normal distribution for the test statistic. This approximation is imprecise, which could lead to an inaccurate estimation of the p-value, consequently resulting in the wrongful classification of sentences not generated by LMs as being LM-produced. For example, given $n=100,\gamma=0.5,L_G(\gamma) = 57$, the p-value of the z-test statistic is about 0.08, indicating that this sentence would be identified as watermarked at 10\% FPR (false positive rate). However, in our case, the p-value is around 0.37, suggesting that we cannot determine this sentence as watermarked. In Table~\ref{tab:test stat comp3}, we compare the empirical FPR of the two test statistics with their theoretical guaranteed FPR on 500 non-watermarked sentences. We can see clearly the empirical FPR is larger than its theoretical guarantee, which validates our assertion that z-test is imprecise on watermark detection. A detailed discussion can be found in Section~\ref{sec:comp test statistic}.

\begin{table}[t]
\centering
\vspace{-0.2cm}
\caption{Comparison of different test statistics on theoretical FPR (false positive rate) and empirical FPR with 500 non-watermarked sentences. We can see clearly the empirical FPR of z-test is continuously greater than its theoretical guarantee.}
\label{tab:test stat comp3}
\scalebox{0.7}{
\begin{tabular}{lcc}
\toprule
 False positive samples/All samples               & $p<0.10$ (10\%FPR)  & $p<0.01$ (1\%FPR) \\ \midrule
z-test \cite{kirchenbauer2023watermark}  & 56/500 (11.2\% FPR)                                      & 12/500 (2.4\% FPR)                   \\
DiPmark statistic & 13/500 (2.6\% FPR)                                     & 4/500  (0.5\% FPR)         \\ \bottomrule
\end{tabular}
}
\vspace{-0.5cm}
\end{table}

\textbf{Detecting efficiency discussion.} Similar to the detection algorithms presented in \cite{kirchenbauer2023watermark}, our watermark detection process is highly efficient, requiring only a single pass through the provided text sequence. However, it is worth noting that the detection algorithm outlined in \citet{kuditipudi2023robust} necessitates iterating through the sequence a staggering 5000 times, which is notably inefficient when compared to our approach. Besides, \citet{hu2023unbiased} requires prompt and language model API during detection, which is also not practical or efficient. A detailed empirical comparison is in Section~$\ref{sec:Efficiency}$.

\begin{figure*}
    \centering
    \includegraphics[width=0.85\textwidth]{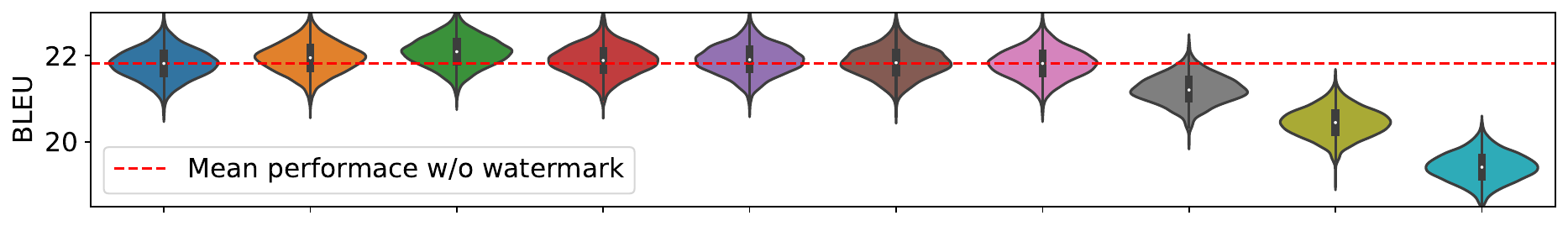}
    \includegraphics[width=0.85\textwidth]{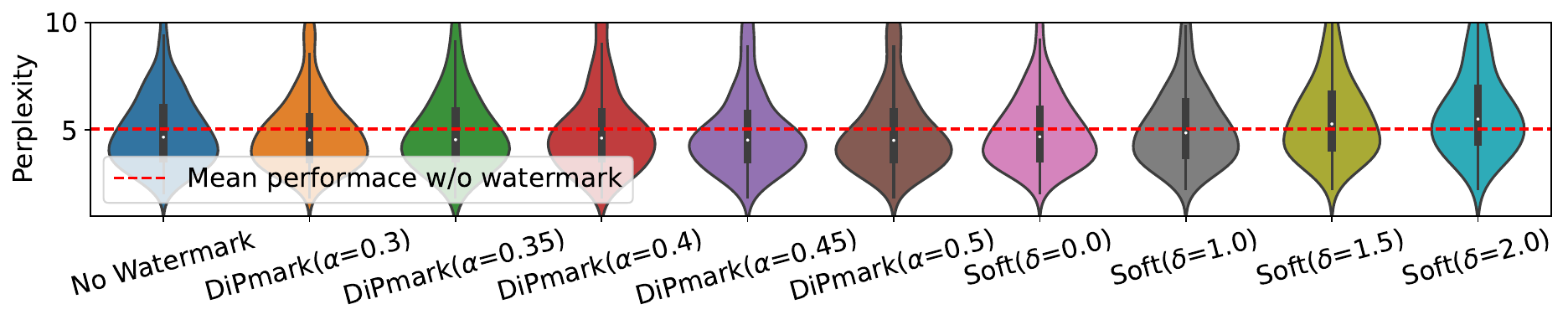}
    \vspace{-13pt}
    \caption{Empirical verification of distribution-preserving property of DiPmark. \textbf{Top:} Violin plot of Machine Translation BLEU. \textbf{Bottom:} Violin plot of Text Summarization Perplexity. We can see the Soft watermarks \cite{kirchenbauer2023watermark} significantly degrade the text quality, while DiPmarks preserve the text quality.}
    \label{fig:violin plot main}
    \vspace{-13pt}
\end{figure*}
\vspace{-0.2cm}
\section{DiPmark is Provably Resilient Against Text Modification}\label{sec:provable robust}
\vspace{-0.1cm}
 In this section, we show that DiPmark possesses provable robustness against arbitrary textual modification attacks with a guaranteed fixed false positive rate. Notably, the existing watermarking approaches are not provable resilient with a guaranteed FPR. \citet{kirchenbauer2023watermark} and \citet{zhao2023provable} assume that the test statistic follows a normal distribution, leading to imprecise guarantee of FPR according to our discussion in Section~\ref{sec:detection}.

\textbf{Problem formulation.} Let $\x_{1:n}$ represent a watermarked sentence. To generate the cipher $\theta$ at the i-th iteration, we employ a hash function $h$, a confidential key $k$, and a texture key $\s:=\x_{i-a:i-1}, a\geq 1$. This indicates that the preceding $a$ tokens serve as the texture key for the watermarking of the token situated at position $i$. During the detection phase, the formula $\Phi(\gamma,\x_{1:n}):=L_G(\gamma)/n-(1-\gamma)$ coupled with a threshold $z$ is applied to ascertain if the text has been watermarked. Notably, within $\Phi(\gamma,\x_{1:n})$, the sole variable associated with textual modification assaults is $L_G(\gamma)$. Consequently, our primary objective is to discern the most severe reduction in $L_G(\gamma)$ for a single token alteration.

\textbf{Worst-case perturbation analysis.} Supposing the token $x_i$ in $\x_{1:n}$ undergoes modification, this will lead to a reduction in $L_G(\gamma)$ through two ways: a) Initially, the token $x_i$ may be categorized as a green token, but post-alteration, it either gets eliminated or transitions into a red token, leading to a potential decline in the number of green tokens $L_G(\gamma)$ by at most 1. b) Since the list of red-green tokens for $x_{i+1},...,x_{i+a}$ is generated by hashing the token $x_i$, its subsequent alteration could cause $x_{i+1},...,x_{i+a}$ to turn into red tokens. In this scenario, the number of green tokens $L_G(\gamma)$ may shrink by a maximum of $a$. As a result, the greatest decline in $L_G(\gamma)$ for a single token modification stands at $a+1$.

\begin{definition}[Certified radius]
    Let $\epsilon\in[0,1]$ denote the fraction of altered tokens. 
    The certified radius of a watermarked sequence is $\epsilon_0$, if for all perturbations confined within the budget $\epsilon\leq\epsilon_0$, the altered watermarked sequence can still be recognized as watermarked.
\end{definition} 

\begin{theorem}\label{thm:provable nrobust}
    Given $\Phi(\gamma,\x_{1:n}):=L_G(\gamma)/n-(1-\gamma)$ and a threshold $z$, the certified radius of the watermarked sequence $\x_{1:n}$ is 
    $\epsilon_0 = \frac{\Phi(\gamma,\x_{1:n})-z}{2+a-\gamma+z}.$
\end{theorem}


\begin{table}[]
\vspace{-0.2cm}
\caption{Distribution-preserving performance of different watermarking methods on machine translation and text summarization. We use F1 scores of BERTScore and scale BERTScore with a factor of 100.}
\label{tab:text quality comp}
\scalebox{0.67}{
\begin{tabular}{l|cc|cc}
\toprule
{}    & \multicolumn{2}{c|}{Machine Translation} & \multicolumn{2}{c}{Text Summarization} \\
\midrule
\textbf{}          & BERTScore$\uparrow$    & BLEU$\uparrow$   & BERTScore$\uparrow$                                     & Perplexity$\downarrow$                                                                      \\
\midrule
\textbf{No Watermark}   & 55.9±0.3           & 21.8±0.3        & 32.73±0.08                                          & 5.021±0.018    \\
\midrule
\textbf{Soft ($\delta$=0.0)}     & 56.0±0.3           & 21.8±0.3        & 32.73±0.08 & 5.021±0.018 \\
\textbf{Soft ($\delta$=1.0)}     & 55.7±0.3           & 21.2±0.3        & 32.37±0.08                                          & 5.309±0.019                                                                                  \\ 
\textbf{Soft ($\delta$=1.5)}     & 55.0±0.3           & 20.4±0.3        & 32.09±0.08                                          & 5.660±0.021                                         \\
\textbf{Soft ($\delta$=2.0)}     & 53.9±0.3           & 19.4±0.3        & 31.46±0.08                                          & 6.241±0.023                                         \\
\midrule
{\textbf{\citet{kuditipudi2023robust}}}    & 56.0±0.3       & 21.7±0.3          & 32.70±0.08          & 5.021±0.021          \\ 
{\textbf{\citet{hu2023unbiased}}}           & 56.3±0.3       & 21.8±0.3          & 32.71±0.08          & 5.023±0.018        \\ 
\midrule
\textbf{DiPmark ($\alpha$=0.3)}  & 56.1±0.3           & 22.0±0.3        & 32.79±0.08                                          & 5.014±0.018                                                                                 \\
\textbf{DiPmark ($\alpha$=0.35)} & 56.2±0.3           & 22.1±0.3        & 32.74±0.08                                          & 4.998±0.018                                                                              \\
\textbf{DiPmark ($\alpha$=0.4)}  & 56.1±0.3           & 21.9±0.3        & 32.77±0.08                                          & 5.001±0.018                                                                               \\
\textbf{DiPmark ($\alpha$=0.45)} & 56.2±0.3           & 21.9±0.3        & 32.69±0.08                                          & 5.024±0.018                                                                             \\
\textbf{DiPmark ($\alpha$=0.5)}  & 56.2±0.3           & 21.8±0.3        & 32.72±0.08                                          & 5.014±0.018                                                                    \\ 
\bottomrule
\end{tabular}
}
\vspace{-0.3cm}
\end{table}

\vspace{-0.2cm}
\section{Experiments}
\vspace{-0.1cm}
Our experimental section consists of five parts. In the first three parts, we compare the distribution-preserving property, accessibility, and resilience of DiPmark with the SOTA watermark methods \citep{kirchenbauer2023watermark,kuditipudi2023robust,hu2023unbiased}. 
In the fourth part, we compare the detectability of DiPmark with the Soft watermark introduced in \cite{kirchenbauer2023watermark}. In the final part, we validate the practicality of DiPmark by conducting a case study on GPT-4~\cite{openai2023gpt}. 
Detailed experimental settings are in Appendix~\ref{sec:detailed_experiment_setup}.


\textbf{General experimental observation.} 
 We find that our DiPmark, configured with $\alpha=0.45$, exhibits comparable levels of detectability and robustness comparing with the Soft watermark ($\delta = 1.5$) \cite{kirchenbauer2023watermark}. Importantly, our DiPmark maintains the same level of text quality as the original language model, owing to its inherent distribution-preserving property.


 \vspace{-0.1cm}
\subsection{{Distribution-preserving Property}}\label{sec:dip perf exp}
 \vspace{-0.1cm}
We will empirically verify the distribution-preserving property of different watermarks. Since DiPmark is \textbf{provably} distribution-preserving (Corollary~\ref{col:disprev}), we use this experiment as a support for the theorem.

We follow the evaluation process of \cite{hu2023unbiased}, where we assess the performance of DiPmark with two seq2seq tasks: text summarization (TS) and machine translation (MT). For the TS task, we employ the BART-large model \citep{liu2020multilingual}. For MT task, we focus on English-to-Romanian translation. We employ the Multilingual BART (MBart) model \citep{liu2020multilingual} on the WMT’14 En-Ro corpus. Specifically for DiPmark, we select values for $\alpha$ from the set $\{0.3, 0.35, 0.4, 0.45, 0.5\}$, while for the Soft watermark \citep{kirchenbauer2023watermark}, we choose green list bias values $\delta$ from the set $\{0.0, 1.0, 1.5, 2.0\}$ alongside a fixed green list separator $\gamma=0.5$, indicating that 50\% of tokens are green while the remainder are red. Notice, Soft watermark with $\delta=0.0$ is equivalent to no watermark since it does not promote the probability of green list tokens.


Upon examining Figure~\ref{fig:violin plot main} and Table~\ref{tab:text quality comp}, we find across all $\alpha$ values in the range $\{0.3, 0.35, 0.4, 0.45, 0.5\}$, the BLEU scores in the machine translation tasks and the perplexity values in the text summarization tasks remain consistently similar between DiPmark and the original language model. However, as we increase the $\delta$ values in the Soft watermark, a notable degradation in text quality becomes evident. A more comprehensive set of results is provided in Appendix~\ref{sec:add1}.
 \vspace{-0.2cm}
\subsection{Accessibility}\label{sec:Efficiency}
 \vspace{-0.1cm}
We compare the time for detecting 1 and 1,000 watermarked sequences with different detection algorithm. The task is text generation with LLaMA-2 (chat, 7B). We use the same GPU (NVIDIA A6000) for all experiments. From Table \ref{tab:efficient} we see the detecting algorithms of DiPmark are efficient without accessing LMs, while \citet{hu2023unbiased} requires additional access to LMs and prompts, and \citet{kuditipudi2023robust} needs significantly longer time.
\begin{table}[h]
\vspace{-0.2cm}
\centering
\caption{{Comparison of accessibility of different watermarks.}}
\label{tab:efficient}
\centering
\scalebox{0.85}{
\begin{tabular}{l|c|c|c}
\toprule
 Number of samples& \textbf{1 } & \textbf{1,000 } &LM \& prompt access\\ \midrule
Soft watermark & 0.3s & 92s & No\\ 
\citet{kuditipudi2023robust} & 80s & 12h & No\\ 
\citet{hu2023unbiased} & 3.4s & 412s & Yes\\ 
DiPmark & 0.3s & 90s & No\\ \bottomrule
\end{tabular}
}
\vspace{-0.5cm}
\end{table}

  \begin{table}[h]
\caption{{AUC score of different watermarks under varying attack strength $\epsilon$ on text generation task. Each row is evaluated over around 500 watermarked and 500 non-watermarked sequences of length n = 260 ± 5.}}
\label{tab:AUCpoem}
\centering
\scalebox{0.88}{
\begin{tabular}{l|cccc}
\toprule
AUC                                                                    & \multicolumn{4}{c}{Random text modification} \\ \midrule
 {}               &  {$\epsilon$ = 0.0} &  {$\epsilon$ = 0.1} &  {$\epsilon$ = 0.2} &  { $\epsilon$= 0.3}                     \\ \midrule
 Soft watermark            &  \textbf{0.9990}  &  \textbf{0.9883}  &  \textbf{0.9521}  &  0.8033        \\
 \citet{kuditipudi2023robust} &  0.9951           &  0.9461           &  0.8979           &  0.7815   \\
 \citet{hu2023unbiased}         &  0.9936           &  0.9297           &  0.8391           &  0.7574    \\
 DiPmark ($\alpha$=0.45)       &  \textbf{0.9990}  &  0.9859           &  0.9515           &  \textbf{0.8060} \\
 \bottomrule
\end{tabular}
}\\
\vspace{0.15cm}
\centering
\scalebox{0.87}{
\begin{tabular}{l|cccc}
\toprule
AUC                                                                    & \multicolumn{4}{c}{Paraphrasing attack}                                                                                                                           \\ \midrule
 {}               & {$\epsilon$ = 0.0}                       & {$\epsilon$ = 0.1}                       & {$\epsilon$ = 0.2}                       & {$\epsilon$ = 0.3}                       \\ \midrule
 Soft watermark &  \textbf{0.9990} &  \textbf{0.9894} &  0.9469          &  0.8157          \\
 \citet{kuditipudi2023robust} &  0.9951          &  0.9529.         &  0.9013          &  0.7711          \\
 \citet{hu2023unbiased} &  0.9936          &  0.9368          &  0.8325         &  0.7661          \\
 DiPmark ($\alpha$=0.45)  &  \textbf{0.9990} &  0.9871          &  \textbf{0.9503} &  \textbf{0.8216}\\
 \bottomrule
\end{tabular}
}
\end{table}

\begin{figure}[h]
    \centering
    \includegraphics[width=0.46\textwidth]{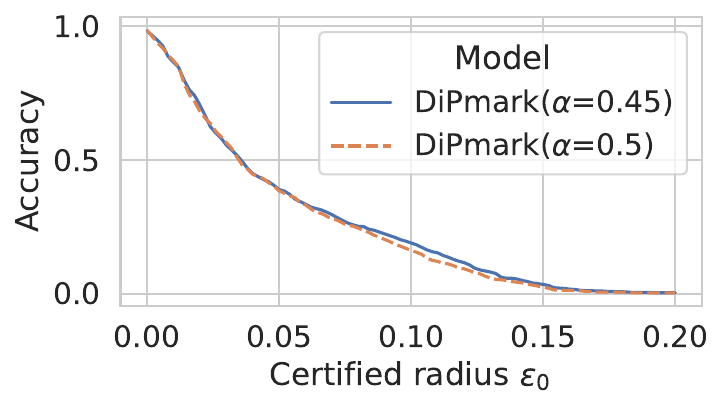}
    \caption{Certified radius $\epsilon_0$ of DiPmark with text modification with FPR smaller than 1\%. The x-axis refers the certified radius and the y-axis refers the percentage of watermarked sequences that are resilience under any text modification attacks with budget $\epsilon_0$.}
    \label{fig:certified}
\end{figure}
\subsection{{Resilience and provable resilience}}\label{sec:robust}
 We compare the resilience of the DiPmark ($\alpha=0.45$) with the SOTA watermark approaches \cite{kirchenbauer2023watermark,kuditipudi2023robust,hu2023unbiased}. In this context, we use the text generation task with 1,000 generated sequences on LLaMA-2. The texture key generation relies on the most recent one token, i.e., $a=1$. For resilience evaluation, we manipulate $\epsilon\in\{0.1, 0.2, 0.3\}$ portion of the text tokens through random text modifications and paraphrasing attacks. We also evaluate the provable resilience of the DiPmark under 1\% FPR, where we use the above mentioned 1,000 generated sequences on LLaMA-2 to calculate the certified radius (Theorem~\ref{thm:provable nrobust}).

 In Table~\ref{tab:AUCpoem}, we report the AUC score of different watermarks under varying attack strength $\epsilon$. The analysis underscores that, when $\epsilon$ remains below $0.3$, DiPmark demonstrates robust performance in effectively detecting watermarked sentences.
 In Figure~\ref{fig:certified}, we also show the certified radius of the watermarked sequences of DiPmark with FPR smaller than 1\% under the text modification.
\subsection{{Ablation study: watermark detectability}}\label{sec:exp detect}
We evaluate the detectability of our watermark on text generation task using LLaMA-2. We generate 1,000 examples for each tasks.
 We select $\alpha \in\{ 0.45, 0.5\}$ for DiPmark, and $\delta \in \{1.0, 1.5, 2.0\}$ and $\gamma=0.5$ for Soft watermark \citep{kirchenbauer2023watermark}. During detection, we use $\gamma=0.5$. We report the Type I (FPR) and II (FNR) errors.  We set the threshold $z=1.073/\sqrt{n}$ (FPR $p\leq0.1$) and $z=1.517/\sqrt{n}$ (FPR $p\leq0.01$). We also report the averaged green token ratio (\ref{def:green token ratio}) vs. text perplexity and token list separator $\gamma$ of DiPmark and Soft watermark. The averaged green token ratio quantifies the bias towards green tokens within the text sequence (see Section~\ref{sec:detection}). Notice, as the z-test in \citet{kirchenbauer2023watermark} is imprecise (see Section~\ref{sec:detection}), we use DiPmark detector for all models. 

\begin{figure}
    \centering
    \includegraphics[width=0.46\textwidth]{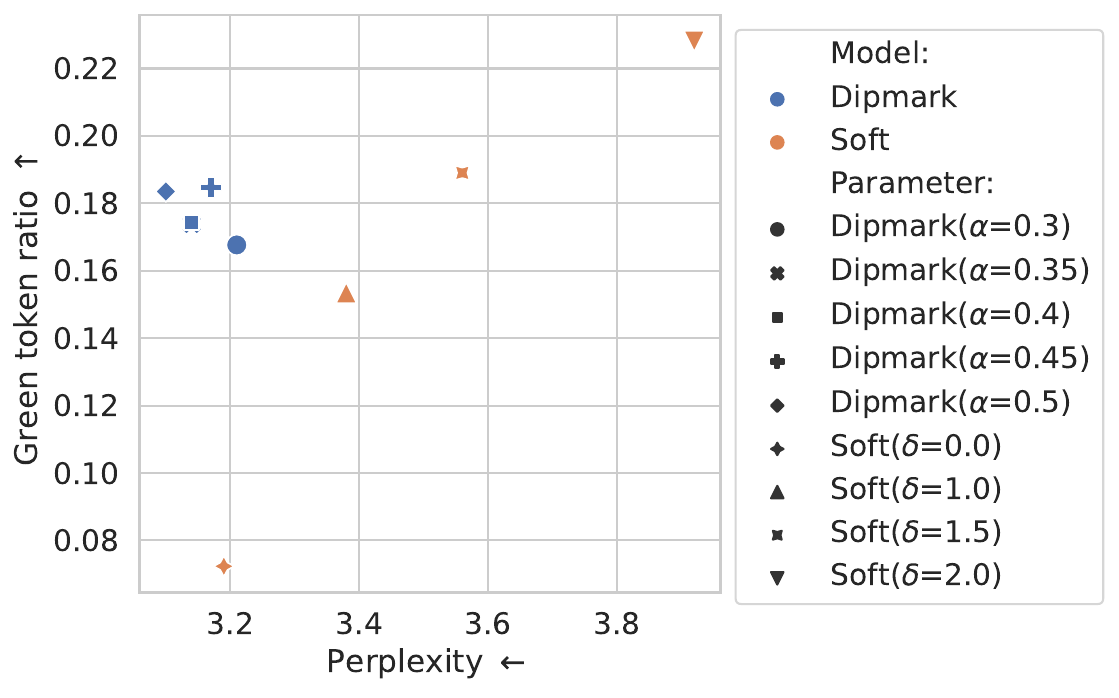}
    \includegraphics[width=0.46\textwidth]{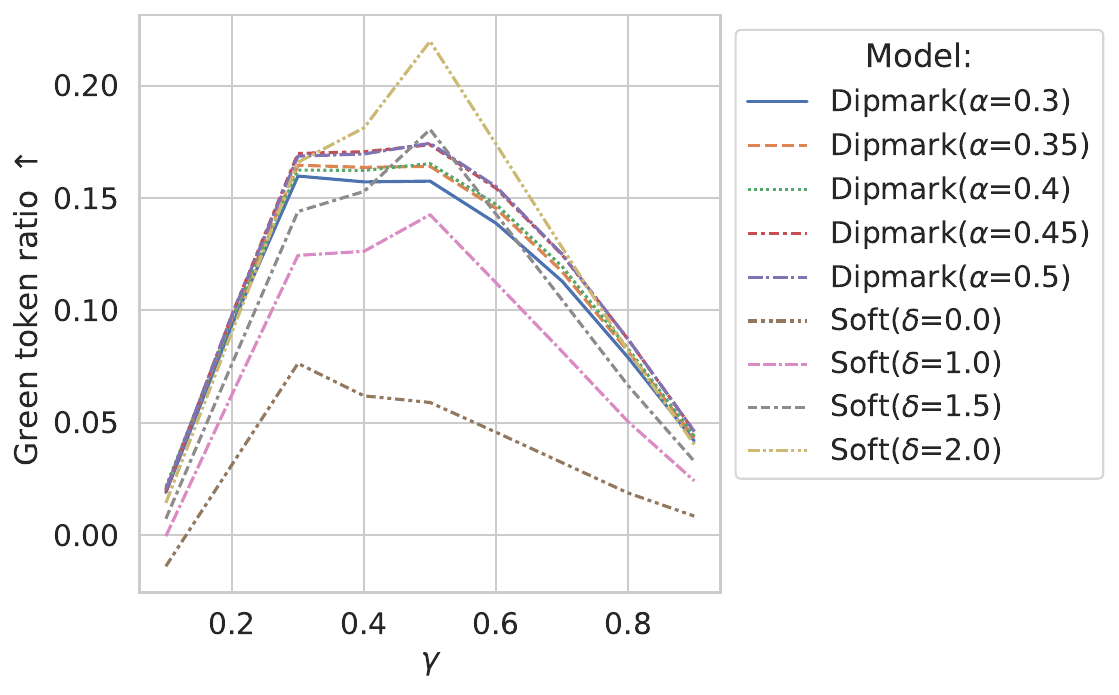}
    \caption{\textbf{Top:} Average perplexity vs green token ratio with $\gamma=0.5$ on text generation tasks. \textbf{Bottom:} Average green token ratio with different $\gamma$.}
    \label{fig:det1}
\end{figure}

\begin{table}[t]
\vspace{-0.15cm}
\caption{Empirical error rates for watermark detection on text generation. Each row is averaged over around 500 watermarked and 500 non-watermarked sequences of length $n = 260 \pm 5$. We select the threshold $z=1.073/\sqrt{n}$ (false positive rate $p\leq0.1$) and $z=1.517/\sqrt{n}$ (false positive rate $p\leq0.01$).}
\label{tab:detection}
    \centering
    \scalebox{0.77}{
\begin{tabular}{l|cccc|c}
\toprule
{}                            & \multicolumn{4}{c}{$z=1.073/\sqrt{n},p\leq0.1$}                                                   \\
\midrule
\textbf{}               & {FPR}$\downarrow$ & {TNR}$\uparrow$        & {TPR}$\uparrow$        & FNR$\downarrow$ & PPL$\downarrow$   \\
\midrule
\textbf{Soft ($\delta$=1.0)}     & 0.0545  & 0.9455 & 0.8919 & 0.2686 &3.38±0.06 \\
\textbf{Soft ($\delta$=1.5)}     & 0.0545  & 0.9455 & 0.9961 & 0.0796 &3.56±0.06\\
\textbf{Soft ($\delta$=2.0)}     & 0.0545  & 0.9455 & 1.0000 & 0.0000 &3.92±0.07\\
\midrule
\textbf{DiPmark ($\alpha$=0.45)} & 0.0545 & 0.9455 & 1.0000 & 0.0000 &3.14±0.06\\
\textbf{DiPmark ($\alpha$=0.5)}  & 0.0545 & 0.9455 & 1.0000 & 0.0000 & 3.17±0.05\\
\bottomrule
\end{tabular}}
\\
\vspace{0.15cm}
\centering
\scalebox{0.77}{
\begin{tabular}{l|cccc|c}
\toprule
{}     & \multicolumn{4}{c}{$z=1.517/\sqrt{n},p\leq0.01$}                                                                             \\
\midrule
\textbf{}    & {FPR}$\downarrow$     & {TNR}$\uparrow$     & {TPR}$\uparrow$        & FNR$\downarrow$ &PPL$\downarrow$   \\
\midrule
\textbf{Soft ($\delta$=1.0)}     & 0.0080 & 0.9920 & 0.8255 & 0.1745 &3.38±0.06\\
\textbf{Soft ($\delta$=1.5)}     & 0.0080 & 0.9920 & 0.9724 & 0.0276 & 3.56±0.06\\
\textbf{Soft ($\delta$=2.0)}     & 0.0080 & 0.9920 & 0.9981 & 0.0019 &3.92±0.07\\
\midrule
\textbf{DiPmark ($\alpha$=0.45)} & 0.0080 & 0.9920 & 0.9794 & 0.0206 &3.14±0.06\\
\textbf{DiPmark ($\alpha$=0.5)}  & 0.0080 & 0.9920 & 0.9827 & 0.0173 & 3.17±0.05 \\
\bottomrule
\end{tabular}
}
\vspace{-0.2cm}
\end{table}

The results for text generation are visually depicted in Figure~\ref{fig:det1}. In Figure~\ref{fig:det1} (top), it is evident that our DiPmark variants with $\alpha=0.45$ and $0.5$ yield green token ratios akin to those of the Soft watermark with $\delta=1.5$ without any discernible degradation in text quality. Figure~\ref{fig:det1} (bottom) delves into the impact of different green list separators $\gamma$, revealing that, for most watermark models, $\gamma=0.5$ yields the highest green token ratio, underscoring its suitability as a reasonable choice for watermark detection. The empirical error rates for watermark detection in text generation are reported in Table~\ref{tab:detection}, showcasing the commendable performance of DiPmark with low false positive rates while maintaining a high true positive rate. Broadly speaking, DiPmark with $\alpha = 0.45$ and $0.5$ exhibit performance comparable to that of the Soft watermark with $\delta = 1.5$ and $2.0$. For more experimental results regarding the detectability, please refer to Appendix \ref{sec:add2}.

\subsection{Case study: watermarking GPT-4 by DiPmark}
Recently, GPT-4 released the log-probability of the top-5 tokens during the generation process. This advancement enables us to modify and apply our DiPmark approach to GPT-4's framework. As we only know the probability of the top-5 tokens, we treat the probability of the rest tokens as 0. Given a prompt, we will first use GPT-4 to generate the top-5 log-probability of the next token. Then we adapt DiPmark to the log-probability and sampling the next token based on the reweighted distribution. Finally, we merge the generated token into the prompt, and repeat the above steps. In our experiments, we use \texttt{gpt-4-0613} on 100 different fiction writing prompts and restrict the number of generated token to 200. We set $\alpha=0.45$ in our DiPmark model.
\begin{figure}[h]
    \centering
    \includegraphics[width=0.48\textwidth]{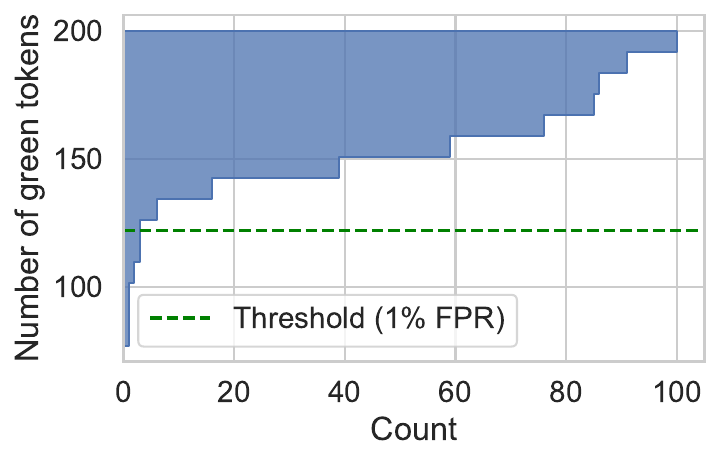}
    \vspace{-12pt}
    \caption{Cumulative histogram of the number of green tokens in the 100 watermarked gpt-4 generated sequences. The green line represents the threshold with FPR smaller than 1\%.}
    \label{fig:gpt4}
    \vspace{-12pt}
\end{figure}

In Figure~\ref{fig:gpt4}, we show the cumulative histogram of the number of green tokens in the 100 watermarked GPT-4 generated sequences. As all generated sequences have 200 tokens, any sequence with greater than 122 green tokens can be detected as watermarked content with FPR less than 1\%. From the plot, we see 97 out of 100 generated sequences can be detected by our algorithm, which validate the applicablity of our watermark on the industry-level LLMs.

\section{Conclusion}
In summary, we present DiPmark, a novel watermarking solution tailored for LLMs. DiPmark exhibits the crucial attributes of distribution-preserving, accessibility, and resilience, which we rigorously substantiate through a combination of theoretical analysis and empirical investigations. 
Our work not only strengthens the theoretical foundations, but also imparts practical insights that are valuable for the industrial deployment of LLM watermarking technologies.

\section*{Impact Statement}
Machine learning holds significant potential to enhance human life, however, its malicious applications could substantially jeopardize safety \cite{wu2022retrievalguard,hong2024improving,pmlr-v202-hu23g,wang2023defending,wang2023distributionally,wu2023law,chen2024your}.
This research focuses on advancing watermark techniques to effectively identify AI-generated sentences. In an era where AI's role in content creation is expanding rapidly, our work gains significance in preserving the authenticity and integrity of digital text. This innovation is pivotal in distinguishing human-authored content from that produced by AI, a distinction that holds substantial value across various societal and technological domains, e.g., enhancing digital content authenticity, combating misinformation, and empowering content creators.

\section*{Acknowledgement}

This work was partially supported by NSF IIS 2347592, 2347604, 2348159, 2348169, DBI 2405416, CCF 2348306, CNS 2347617; HY Zhang was supported by NSERC Discovery Grant RGPIN-2022-03215, DGECR-2022-00357.


\bibliography{example_paper.bib}
\bibliographystyle{icml2024}

\newpage
\appendix
\onecolumn
\section{Future Work}

Future endeavors should focus on enhancing the detectability of distribution-preserving watermarks. This could be realized by assigning greater weight to the green-list tokens during the watermarking process. Additionally, a promising avenue for exploration involves the design of a more robust distribution-preserving watermark, potentially through the integration of multiple detectors. These directions represent promising opportunities for advancing the efficacy and applicability of watermarking techniques on large language models.

\section{Related Work}\label{app:related work}
\textbf{Reweight-based watermarking framework.} 
In a recent seminal work, \cite{kirchenbauer2023watermark} introduced a pioneering watermarking scheme tailored for LLMs, backed by formal guarantees. Their work demonstrated that watermark embedding could be accomplished by altering the token distribution during generation, targeting outputs with substantial entropy. However, this approach inevitably leads to a pivotal change in the distribution of the generated text, potentially compromising the quality of the generated content.

To maintain an unaltered output distribution in watermarked content, alternative strategies have been explored. \cite{christ2023undetectable} and \cite{kuditipudi2023robust} employed the inverse sampling method to generate watermarked token distributions. Notably, \cite{christ2023undetectable}'s method faces resilience issues under modifications and lacks empirical validation for detectability. Meanwhile, \cite{kuditipudi2023robust}'s approach necessitates the secret key distribution during detection, potentially compromising data security and watermark stealthiness. Moreover, their detection process involves hundreds of resampling steps from the secret key distribution, which is inefficient for lengthy texts. \cite{hu2023unbiased} used inverse sampling and permutation based reweight methods for watermarking, but the detector requires access of the language model API, undermining its operational efficiency. Aaronson's ongoing watermarking project \citep{Aaronson2022} employs n-gram hashing for reweighting the next-token distribution, though specific details are currently unavailable.

The landscape also includes several schemes \citep{abdelnabi2021adversarial,qiang2023natural,yoo2023robust,munyer2023deeptextmark} that incorporate an ML model within the watermarking algorithm itself. However, these constructions lack formal assurances and rely on heuristic arguments for satisfying the criteria of Stealthiness, Efficiency, and Resilience.

Our research aligns closely with the findings presented in \cite{kirchenbauer2023watermark}. In their methodology, they employed watermarking for text derived from a language model by bifurcating the token set into designated `red' and `green' lists. The division is determined by a random seed that is contingent on the secret key coupled with a hash of priorly generated tokens. The authors accentuated the prominence of green tokens during the sampling phase by reweighting the token log-probabilities. Building on this foundation, our research retains the red-green list configuration, but introduces an evolved family of permutation-based reweight strategies. This dual approach ensures: 1) a promoted utilization of green tokens, and 2) equivalency in distribution between a sample from the watermarked language model and one from the original language model. 



\textbf{Post-hoc detectors.}
Post-hoc detection stands as a notable alternative to watermarking, focusing on the retrospective analysis of machine-generated text. This could be achieved through leveraging features inherent to language models or by refining pre-existing, expansive language models to function as detectors, as elaborated by \citep{zellers2019defending}. Notably, specific implementation nuances, such as sampling methodologies, can be discerned through reverse engineering the generated text, a process detailed by \citep{tay2020reverse}.  There are also post-hoc detectors designed for the modern large language models \citep{mitchell2023detectgpt,tian2023gptzero,kirchner2023new}, which are models specifically trained for the binary detection task. 
However, there is a growing sentiment that those detection methodologies are diminishing in efficacy in tandem with the evolution of language model capabilities. As \cite{gambini2022pushing} observed, detection mechanisms that were adept with GPT-2 have encountered challenges with GPT-3. Besides, the text rephrasing model in \citep{krishna2023paraphrasing} bypassing prevalent post-hoc detectors like GPTZero \citep{tian2023gptzero}, DetectGPT \citep{mitchell2023detectgpt}, and OpenAI's proprietary detector \citep{kirchner2023new}. Additionally, a pertinent observation made by \cite{chakraborty2023possibilities} suggests that as AI-generated content becomes increasingly indistinguishable from human-produced text, the demands on post-hoc detectors to analyze more extended text segments will escalate.

\textbf{Steganography.}
Steganography involves embedding concealed messages in channels such as natural language or images, ensuring only intended recipients can discern the message while others remain unaware \citep{hopper2002provably}. When applied to watermarking, the aim is stealthy. Yet, known steganography techniques might not achieve this without certain entropy-related assumptions. In scenarios where language model prompts can be chosen adversarially, the need for stealthy persists. This discrepancy arises due to differences in access levels that watermarking and steganography have to the model's output distribution. In steganography, there's only oracle access to this distribution. Conversely, our watermarking approach gets a detailed view of the token's probability distribution. Hence, while steganography either relies on entropy assumptions \citep{hopper2002provably} or compromises security with low entropy channels \citep{dedic2009upper}, our watermark remains stealthy irrespective of the text's entropy. This is achieved by leveraging the full distribution access and using it as a foundation for embedding watermarks. \cite{kaptchuk2021meteor} offers encoding similar access. However, it presupposes equal decoding access, which is impractical for watermarking as the detection algorithm won't typically have the initiating prompt, thus remaining ignorant of the distribution.
\section{Missing Proofs}\label{sec:missing proof}
\subsection{Proof of Theorem \ref{thm:dipreweight}}
\begin{proof} We need to show $\forall t\in V,\bbE_{\theta}[P_W(t|\x,\theta)] =P_M(t|\x)$. Recall $\theta$ is uniformly distributed on $\Theta$, we have
  \begin{equation}
    \begin{split}
        \bbE_{\theta\sim P_\Theta}[P_W(t|\x,\theta)] &= \sum_{V^p\in\Theta}\bbE_{\theta\sim P_\Theta}[P_W(t|\x,V^p)\bm{1}_{\theta=V^p}]\\
        &= \sum_{V^p\in\Theta}[P_W(t|\x,V^p)]\bbE_{\theta\sim P_\Theta}[\bm{1}_{\theta=V^p}]\\
        &=\frac{1}{N!}\sum_{V^p\in\Theta}P_W(t|\x,V^p).
    \end{split}
\end{equation}
 Given an token $t$ and a permutation of the token list $V^p$, denote by $E_{V^p}(t)$ the position of $t$ in the ordered token set $V^p$.
Let $V^{p^r}$ be the reversed permutation of $V^p$, notice $t$ is the $(N+1-E_{V^p}(t))$-th element in $V^{p^r}$. Given an arbitrary permutation pair $(V^p,V^{p^r})$, $V^p:=\{t_1,...,t_N\}$. We will show
$$P_W(t|\x,V^{p})+P_W(t|\x,V^{p^r})=2P_M(t|\x).$$
For the ease of notation we denote by $i=E_{V^p}(t)$, we have $t_i=t$. From the definition of DiP-reweight we know $P_W(t|\x,V^{p}) = F(E_{V^p}(t)|V^p)-F(E_{V^p}(t)-1|V^p) = F(i|V^p)-F(i-1|V^p)$, where 

\begin{equation}
    F(i|V^p):=\max\left\{\sum_{j=1}^{i} P_M(t_j|\x)-\alpha,0\right\}
    +\max\left\{\sum_{j=1}^{i} P_M(t_j|\x)-(1-\alpha),0\right\},\ i\in[1,N],
\end{equation}
So we need to show
$$F(i|V^p)-F(i-1|V^p)+F(N+1-i|V^{p^r})-F(N-i|V^{p^r})=2P_M(t|\x).$$

As $\sum_{j=1}^{N}P_M(t_j|\x)=1$, we have
\begin{equation}
    \begin{split}
        F(N+1-i|V^{p^r}) &= \max\left\{\sum_{j=1}^{N+1-i} P_M(t_{N+1-j}|\x)-\alpha,0\right\}
    +\max\left\{\sum_{j=1}^{N+1-i} P_M(t_{N+1-j}|\x)-(1-\alpha),0\right\}\\
    &= \max\left\{\sum_{j=i}^{N} P_M(t_{j}|\x)-\alpha,0\right\}
    +\max\left\{\sum_{j=i}^{N} P_M(t_{j}|\x)-(1-\alpha),0\right\}\\
    &= \max\left\{(1-\alpha)-\sum_{j=i}^{i-1} P_M(t_{j}|\x),0\right\}
    +\max\left\{\alpha-\sum_{j=1}^{i-1} P_M(t_{j}|\x),0\right\},
    \end{split}
\end{equation}
and 
\begin{equation}
    \begin{split}
        F(i-1|V^{p}) &= \max\left\{\sum_{j=1}^{i-1} P_M(t_{j}|\x)-\alpha,0\right\}
    +\max\left\{\sum_{j=1}^{i-1} P_M(t_{j}|\x)-(1-\alpha),0\right\}.
    \end{split}
\end{equation}
By $(\max\{A,0\}-\max\{-A,0\})=A,\forall A\in\bbR$, we have 
\begin{equation}
    \begin{split}
        F(N+1-i|V^{p^r})-F(i-1|V^{p}) &= (1-\alpha)-\sum_{j=i}^{i-1} P_M(t_{j}|\x) + \alpha-\sum_{j=1}^{i-1} P_M(t_{j}|\x)\\
        &=1-2\sum_{j=i}^{i-1} P_M(t_{j}|\x).
    \end{split}
\end{equation}
Analogously, we have 
\begin{equation}
    \begin{split}
        F(N-i|V^{p^r})-F(i|V^{p})= 1-2\sum_{j=i}^{i} P_M(t_{j}|\x).
    \end{split}
\end{equation}
Thus, 
\begin{equation}
    \begin{split}
    P_W(t|\x,V^{p})+P_W(t|\x,V^{p^r})=&F(i|V^p)-F(i-1|V^p)+F(N+1-i|V^{p^r})-F(N-i|V^{p^r})\\
    =&(1-2\sum_{j=i}^{i-1} P_M(t_{j}|\x))-(1-2\sum_{j=i}^{i} P_M(t_{j}|\x))\\
    =&2P_M(t_i|\x) = 2P_M(t|\x).
    \end{split}
\end{equation}

By the symmetric of permutation we have
\begin{equation}
    \begin{split}
    2\bbE_{\theta\sim\Theta}[P_W(t|\x,\theta)] =& \frac{1}{ N!}\sum_{V^p\in\Sigma}P_W(t|\x,V^{p})\\
    =&\frac{1}{ N!}\sum_{V^p\in\Sigma}[P_W(t|\x,V^{p})+P_W(t|\x,V^{p^r})]\\
    =&\frac{1}{ N!}\sum_{V^p\in\Sigma}2P_M(t|\x)\\
    =&2P_M(t|\x).
        \end{split}
\end{equation}
Therefore, $\bbE_{\theta\sim\Theta}[P_W(t|\x,\theta)] = P_M(t|\x)$, which concludes the proof.
\end{proof}
\subsection{Proof of Theorem~\ref{thm:concenbound}}
\begin{proof}
    As $L_G(\gamma) = \sum_{i=1}^n B_i(\gamma)$, where $B_i(\gamma)\sim\textrm{Bernoulli}(1-\gamma)$. By Markov's inequality we have $\forall h>0$,
    $$\Pr(L_G(\gamma)-(1-\gamma)n\geq nt)\leq \frac{\bbE[e^{h(L_G(\gamma)-(1-\gamma)n)}]}{e^{hnt}},$$
    as $B_i$ is independent from each other, we have
    $$\frac{\bbE[e^{h(L_G(\gamma)-(1-\gamma)n)}]}{e^{hnt}} = \prod_{i=1}^n\frac{\bbE[e^{h(B_i-(1-\gamma))}]}{e^{ht}}.$$
    Since $B_i$ follows Bernoulli distribution, we have 
    $$\bbE[e^{h(B_i-(1-\gamma))}]/e^{ht} = (1-\gamma)e^{h(\gamma-t)}+\gamma e^{-h(1-\gamma+t)}.$$
    Thus \begin{equation}\label{eq:xdd}
        \Pr(L_G(\gamma)-(1-\gamma)n\geq nt)\leq [(1-\gamma)e^{h(\gamma-t)}+\gamma e^{-h(1-\gamma+t)}]^n
    \end{equation}
    holds for arbitrary $h>0$.
    Denote by $m(h) = (1-\gamma)e^{h(\gamma-t)}+\gamma e^{-h(1-\gamma+t)}$, taking derivative w.r.t. $h$ yields
    $$\frac{dm(h)}{dh} = (1-\gamma)(\gamma-t)e^{h(\gamma-t)} + \gamma(1-\gamma+t) e^{-h(1-\gamma+t)}.$$
    Let $\frac{dm(h)}{dh} = 0$, we have $h = \ln\frac{\gamma(1-\gamma+t)}{(1-\gamma)(\gamma-t)}.$ Combining it with \autoref{eq:xdd} yields
    \begin{equation}
        \begin{split}
            \Pr(L_G(\gamma)-(1-\gamma)n\geq nt)&\leq \inf_{h>0} [(1-\gamma)e^{h(\gamma-t)}+\gamma e^{-h(1-\gamma+t)}]^n\\
            &\leq [e^{(\gamma-t)\ln\frac{\gamma(1-\gamma+t)}{(1-\gamma)(\gamma-t)}}(1-\gamma+\gamma e^{-\ln\frac{\gamma(1-\gamma+t)}{(1-\gamma)(\gamma-t)}})]^n\\
            & = [e^{(\gamma-t)\ln\frac{\gamma(1-\gamma+t)}{(1-\gamma)(\gamma-t)}}\frac{1-\gamma}{1-\gamma+t}]^n\\
            & = [e^{(\gamma-t)\ln\frac{\gamma(1-\gamma+t)}{(1-\gamma)(\gamma-t)}+\ln\frac{1-\gamma}{1-\gamma+t}}]^n\\
            & = e^{-n((1-\gamma+t)\ln\frac{1-\gamma+t}{1-\gamma}+(\gamma-t)\ln\frac{\gamma-t}{\gamma})}\\
            & = e^{-n\mathbb{KL}(1-\gamma+t||1-\gamma)}
        \end{split}
    \end{equation}

\end{proof}
\subsection{Proof of Theorem~\ref{thm:provable nrobust} and discussion}
\begin{proof}
    Notice based on above discussion, the worst-case decrease on $L_G(\gamma)$ per token modification is $a+1$. If we are allowed to perturbed $\epsilon$ portion of the text, the worst-case decrease on $L_G(\gamma)$ will be $(a+1)\epsilon n$. Denoted by $\x_{1:n'}$ the perturbed text. Assume we can still correctly detect the watermarked sequence, which means 
    $$(L_G(\gamma)-(a+1)\epsilon n)/n'-(1-\gamma)\geq z.$$
    Notice, the left hand side of the above equation is decreasing with $n'$, as we perturbed $\epsilon$ portion of the text, the maximum of the possible $n'$ is $n' = (1+\epsilon)n$, i.e., all modifications are text insertion. In this case, we need to solve 
    $$\frac{L_G(\gamma)-(a+1)\epsilon n}{(1+\epsilon)n}-(1-\gamma)\geq z.$$
    we have 
    $$\epsilon\leq \frac{L_G(\gamma)-(1-\gamma)n-zn}{(2+a-\gamma+z)n}.$$
    
    Therefore, for any text modification with budget $\epsilon\leq \frac{L_G(\gamma)-(1-\gamma)n-zn}{(2+a-\gamma+z)n}$, our algorithm can still detect the watermarked sequence.
 \end{proof}   
    In the following theorem, we provide a more simple certified radius assuming the text length is not changed by perturbations.

\begin{theorem}\label{thm:provable nrobust1}
    Assuming the sequence length $n$ is not changed through text modifications. Given $\Phi(\gamma,\x_{1:n}):=L_G(\gamma)/n-(1-\gamma)$ and a threshold $z$, the certified radius of the watermarked sequence $\x_{1:n}$ is 
    $\epsilon_0 = \frac{\Phi(\gamma,\x_{1:n}) -z}{a+1}$.
\end{theorem}
\begin{proof}
    Notice based on above discussion, the worst-case decrease on $L_G(\gamma)$ per token modification is $a+1$. If we are allowed to perturbed $\epsilon$ portion of the text, the worst-case decrease on $L_G(\gamma)$ will be $(a+1)\epsilon n$. Assume we can still correctly detect the watermarked sequence, which means 
    $$(L_G(\gamma)-(1-\gamma)n-(a+1)\epsilon n)/\sqrt{n}\geq z,$$
    we have 
    $\epsilon\leq \frac{\Phi(\gamma,\x_{1:n}) -z}{(a+1)\sqrt{n}}.$
    Therefore, for any text modification with budget $\epsilon\leq \frac{\Phi(\gamma,\x_{1:n}) -z}{(a+1)\sqrt{n}}$, our algorithm can still detect the watermarked sequence.

\end{proof}

\section{Comparison of the test statistic}\label{sec:comp test statistic}
In this section, we provide a detailed comparison of our test statistic and the z-test statistic proposed in \cite{kirchenbauer2023watermark}. In Figure~\ref{fig:test stat comp1}, we show number of green tokens vs p-value (false positive rate), where we set the number of tokens $n=200$, green list separator $\gamma=0.5$. We see that given the same number of green tokens, the z-test statistic always leads to lower p-value than DiPmark test statistic. Given the fact that the z-test statistic is only an approximation of the green token distribution, we conclude that this approximation is not proper for watermark detection, as it will wrongly classify the sentences not generated by LMs as being LM-produced. In Table~\ref{tab:test stat comp2}, we show the detecting result based on DiPmark detector and the detector in \citet{kirchenbauer2023watermark} on 500 non-watermarked sentences with length 260. We can see clearly the empirical FPR of z-test is continuously greater than its theoretical guarantee, which indicates z-test statistic may not be suitable for watermark detection.
\begin{figure}
    \centering
    \includegraphics[width=0.22\textwidth]{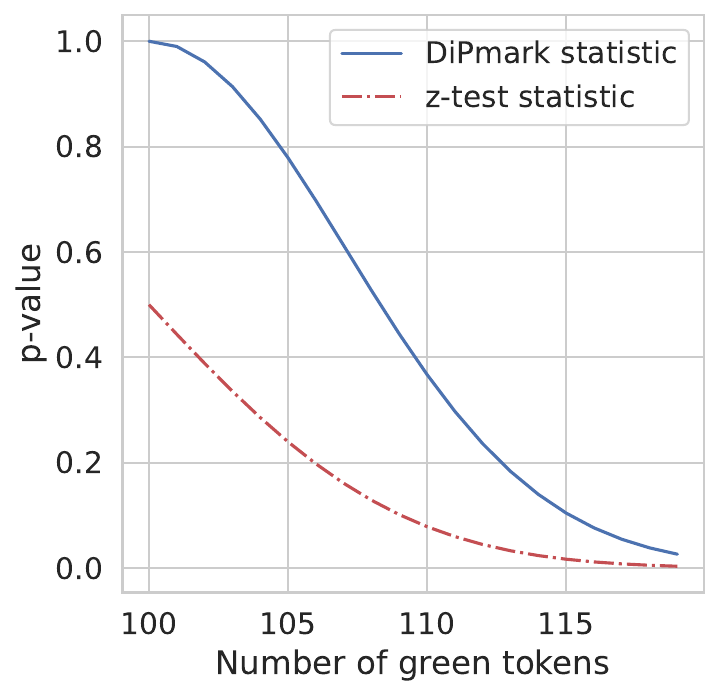}
    \includegraphics[width=0.24\textwidth]{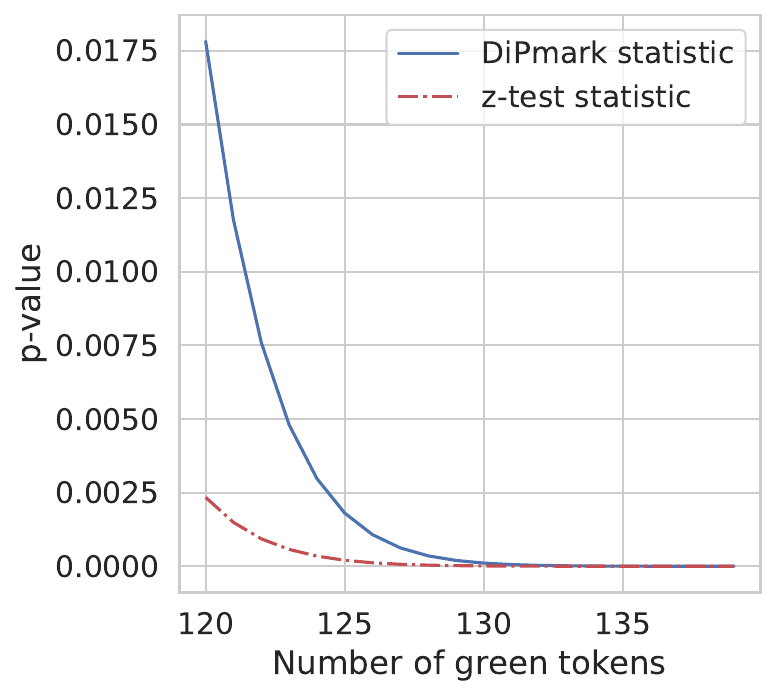}
    \includegraphics[width=0.22\textwidth]{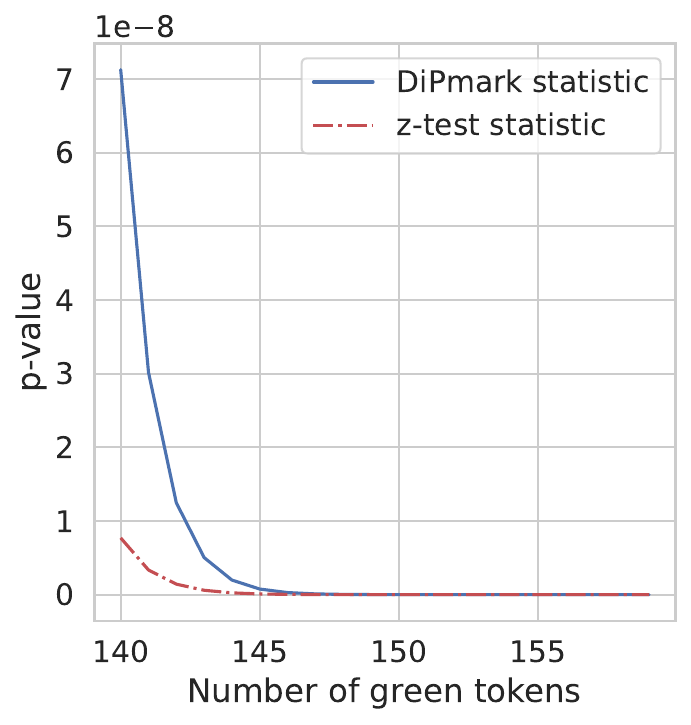}
    \includegraphics[width=0.24\textwidth]{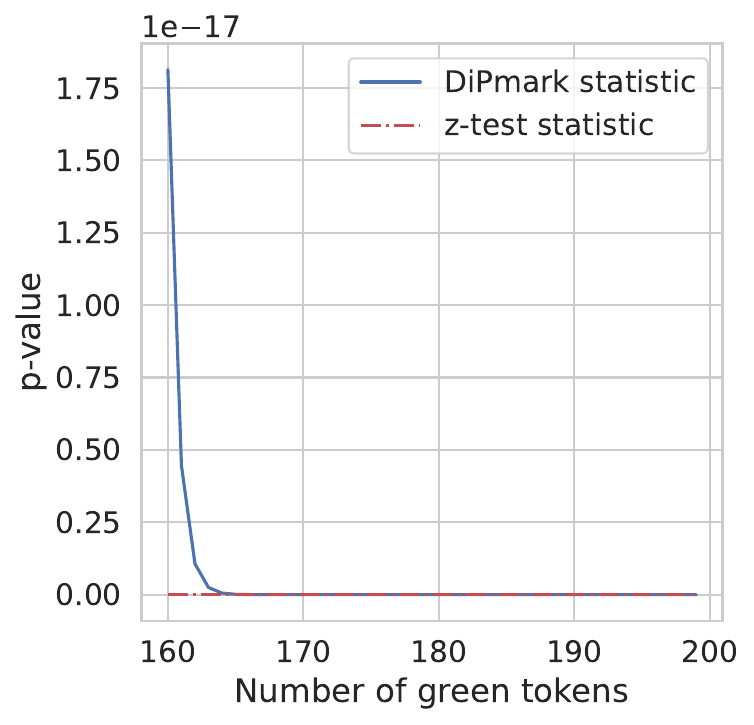}
    \vspace{-0.5cm}
    \caption{Number of green tokens vs p-value (false positive rate), where we set the number of tokens $n=200$, green list separator $\gamma=0.5$. We see that given the same number of green tokens, the z-test always has lower p-value than DiPmark test statistic. Given the fact that the z-test statistic is only an approximation of the green token distribution, we conclude that this approximation is not proper for watermark detection, as it will wrongly classify the sentences not generated by LMs as being LM-produced.}
    \label{fig:test stat comp1}
    \vspace{-0.3cm}
\end{figure}

\begin{table}[]
\centering
\caption{Comparison of test statistics: Theoretical FPR vs Empirical FPR. We can see clearly the empirical FPR of z-test is continuously greater than its theoretical guarantee, which indicates z-test statistic may not be suitable for watermark detection.}
\label{tab:test stat comp2}
\begin{tabular}{lccc}
\toprule
                  & $p<0.10$ (10\%FPR) & $p<0.05$ (5\%FPR) & $p<0.01$ (1\%FPR) \\ \midrule
z-test \cite{kirchenbauer2023watermark}  & 56/500 (11.2\% FPR)                    & 34/500 (6.8\% FPR)                   & 12/500 (2.4\% FPR)                   \\
DiPmark statistic & 13/500 (2.6\% FPR)                    & 10/500  (2\% FPR)                  & 4/500  (0.5\% FPR)         \\ \bottomrule
\end{tabular}
\end{table}

\section{Detailed Experiment Setup}\label{sec:detailed_experiment_setup}

We assess the performance of DiPmark across three critical applications of seq2seq models: text summarization, machine translation, and text generation. The experiments are implemented using the Huggingface library \citep{wolf2019huggingface}, a widely adopted platform for model development and sharing within the NLP community. All experiments are conducted on three Nvidia A6000 GPUs with 48GB of memory. Detecting 1,000 watermarked sentences generated from LLaMA-2 requires only 90 seconds.

\textbf{Machine Translation.} For the machine translation task, we utilize the WMT'14 English (En) to Romanian (Ro) dataset, comprising 1,999 examples in the test set. We employ the Multilingual Bart (MBart) model \citep{liu2020multilingual} along with its official tokenizer.

\textbf{Text Summarization.} In the text summarization task, we use the test set from the CNN-DM corpus \citep{hermann2015teaching}, consisting of 11,490 examples. Our model of choice is BART-large, which encompasses 400 million parameters, and LLaMA-2 with 7 billion parameters.

\textbf{Text Generation.} For text generation, we incorporate the test set from the CNN-DM corpus as part of the generation prompt. We use LLaMA-2 which has 7 billion parameters.


\textbf{Watermark Setup.} Our experiments primarily compare DiPmark with the Soft watermark introduced by \citep{kirchenbauer2023watermark}. In the case of DiPmark, we consider various values of $\alpha$ from the set $\{0.3, 0.35, 0.4, 0.45, 0.5\}$. For the Soft watermark \citep{kirchenbauer2023watermark}, we explore green list bias $\delta$ values from $\{0.0, 1.0, 1.5, 2.0\}$ with a fixed green list separator $\gamma=0.5$.
Texture key generation relies on the most recent five tokens as texture key. For instance, when generating $x_4$ in response to $(x_1, x_2, x_3)$ as the current input to the decoder, the texture key includes $(x_1, x_2, x_3)$, considering the availability of only three tokens. The texture key history resets before generating each batch. To generate the cipher, we employ SHA-256 as the hash function and a set of 1024-bit random bitstrings as the key set $K$. The cipher $\theta$ is sampled from $\Theta$ using $\text{hash}(k,\s)$ as the random seed. We compare DiPmark with ITS \cite{kuditipudi2023robust} and $\delta$-watermark \cite{hu2023unbiased}, where we follow the setting in their open sourced code\footnote{\url{https://github.com/jthickstun/watermark}}\footnote{\url{https://github.com/xiaoniu-578fa6bff964d005/UnbiasedWatermark}}.

\textbf{Evaluation metrics for text quality.} In this part, we introduce the evaluation metrics we used for evaluating the text quality (Section.~\ref{sec:dip perf exp}). 
\begin{itemize}
    \item \textbf{ROUGE score.} For the summarization task, we utilize the ROUGE score \citep{lin2004rouge}, which measures n-gram overlap to assess the summary's effectiveness in capturing essential content from reference summaries.
    \item \textbf{BLEU score.} For the machine translation task, we rely on the BLEU score \citep{papineni2002bleu}, emphasizing the lexical similarity between machine-generated translations and human reference translations.
    \item \textbf{BERTScore.} BERTScore \cite{zhang2019bertscore} computes the similarity of two sentences as a sum of cosine similarities between their tokens' embeddings. We use BERTScore-F1, BERTScore-Precision, and BERTScore-Recall for evaluating both text summarization and machine translation tasks.
    
    \item \textbf{Perplexity.} In information theory, perplexity is a measurement of how well a probability distribution or probability model predicts a sample. It may be used to compare probability models. A low perplexity indicates the probability distribution is good at predicting the sample. We use perplexity for evaluating both text summarization and machine translation tasks.
\end{itemize}
\textbf{Evaluation metrics for detectability of watermarks.} In this part, we introduce the evaluation metrics we used for evaluating the detectability of watermarks (Sections~\ref{sec:exp detect} and \ref{sec:robust}). 
\begin{itemize}
    \item \textbf{Green token ratio.} Denoted by $L_G(\gamma)$ the number of green tokens in a text sequence with green list separator $\gamma$. The green token ratio is given by $L_G(\gamma)/n-(1-\gamma)$. This ratio quantifies the bias towards green tokens within the text sequence (see Section~\ref{sec:detection}).
    \item \textbf{z-score.} The z-score of a text sequence $\x_{1:n}$ is $(L_G(\gamma)-(1-\gamma)n)/\sqrt{n}$. A higher z-score will reduce the false positive rate, where a non-watermarked sequence is detected as watermarked (see Section~\ref{sec:detection}). 
    \item \textbf{Type I and II errors.} We generally use true positive rate (TPR), false positive rate (FPR), true negative rate (TNR), and false negative rate (FNR) to evaluate the performance of watermarks on a mixture of watermarked and non-watermarked sentence. FPR measures the Type I error of the hypothesis testing, in which the null hypothesis got rejected when it is actually true. FNR measures the type II error, in which one fails to reject a null hypothesis that is actually false.
\end{itemize}

\section{Additional Experiments}
\begin{table}[]
\centering
\caption{Performance of Machine Translation.}
\label{tab:machineadd}
\begin{tabular}{l|cccc}
\toprule
\textbf{}                                                        & BERT-F1 & BERT-Precision & BERT-Recall & BLEU \\
\midrule
No Watermark                                                 & 0.559±0.003           & 0.545±0.004                  & 0.574±0.003               & 21.8±0.3      \\
\midrule
DiPmark($\alpha$=0.3)  & 0.561±0.003           & 0.547±0.004                  & 0.575±0.003               & 22.0±0.3      \\
DiPmark($\alpha$=0.35) & 0.562±0.003           & 0.548±0.004                  & 0.575±0.003               & 22.1±0.3      \\
DiPmark($\alpha$=0.4)  & 0.561±0.003           & 0.547±0.004                  & 0.576±0.003               & 21.9±0.3      \\
DiPmark($\alpha$=0.45) & 0.562±0.003           & 0.548±0.004                  & 0.576±0.003               & 21.9±0.3      \\
DiPmark($\alpha$=0.5)  & 0.562±0.003           & 0.548±0.004                  & 0.576±0.003               & 21.8±0.3      \\ \midrule
Soft($\delta$=0.0)     & 0.560±0.003           & 0.545±0.004                  & 0.574±0.003               & 21.8±0.3      \\
Soft($\delta$=1.0)    & 0.557±0.003           & 0.543±0.004                  & 0.572±0.003               & 21.2±0.3      \\
Soft($\delta$=1.5)     & 0.550±0.003           & 0.534±0.004                  & 0.565±0.003               & 20.4±0.3      \\
Soft($\delta$=2.0)     & 0.539±0.003           & 0.523±0.004                  & 0.555±0.003               & 19.4±0.3     \\
\bottomrule
\end{tabular}
\end{table}
\begin{table}[]
\caption{Performance of Text Summarization.}
\centering
\label{tab:summaradd}
\scalebox{0.675}{
\begin{tabular}{l|ccccccc}
\toprule
\textbf{}          & BERT-F1 & BERT-Precision & BERT-Recall & Perplexity & Rouge-1 & Rouge-2 & Rouge-L \\
\midrule
No Watermark            & 0.3273±0.0008         & 0.3181±0.0009                & 0.3366±0.0010             & 5.021±0.018  & 0.3855±0.0009   & 0.1387±0.0008   & 0.2444±0.0008   \\ \midrule
DiPmark($\alpha$=0.3)         & 0.3279±0.0008         & 0.3187±0.0009                & 0.3372±0.0010             & 5.014±0.018  & 0.3861±0.0009   & 0.1390±0.0008   & 0.2450±0.0008   \\
DiPmark($\alpha$=0.35)        & 0.3274±0.0008         & 0.3183±0.0009                & 0.3367±0.0010             & 4.998±0.018  & 0.3856±0.0009   & 0.1389±0.0008   & 0.2449±0.0008   \\
DiPmark($\alpha$=0.4)         & 0.3277±0.0008         & 0.3187±0.0009                & 0.3370±0.0010             & 5.001±0.018  & 0.3862±0.0009   & 0.1392±0.0008   & 0.2449±0.0007   \\ 
DiPmark($\alpha$=0.45)        & 0.3269±0.0008         & 0.3178±0.0009                & 0.3361±0.0010             & 5.024±0.018  & 0.3852±0.0009   & 0.1391±0.0008   & 0.2447±0.0008   \\
DiPmark($\alpha$=0.5)         & 0.3272±0.0008         & 0.3181±0.0009                & 0.3364±0.0010             & 5.014±0.018  & 0.3859±0.0009   & 0.1396±0.0008   & 0.2450±0.0008   \\ \midrule
Soft($\delta$=0.0)            & 0.3273±0.0008         & 0.3181±0.0009                & 0.3366±0.0010             & 5.021±0.018  & 0.3855±0.0009   & 0.1387±0.0008   & 0.2444±0.0008   \\
Soft($\delta$=1.0)            & 0.3237±0.0008         & 0.3137±0.0009                & 0.3338±0.0009             & 5.309±0.019  & 0.3816±0.0009   & 0.1348±0.0008   & 0.2411±0.0007   \\
Soft($\delta$=1.5)            & 0.3209±0.0008         & 0.3097±0.0009                & 0.3323±0.0010             & 5.660±0.021  & 0.3793±0.0009   & 0.1317±0.0007   & 0.2379±0.0007   \\
Soft($\delta$=2.0)            & 0.3146±0.0008         & 0.3027±0.0009                & 0.3266±0.0009             & 6.241±0.023  & 0.3725±0.0009   & 0.1252±0.0007   & 0.2321±0.0007 \\
\bottomrule
\end{tabular}
}
\end{table}

\begin{figure}
    \centering
     \includegraphics[width=1\textwidth]{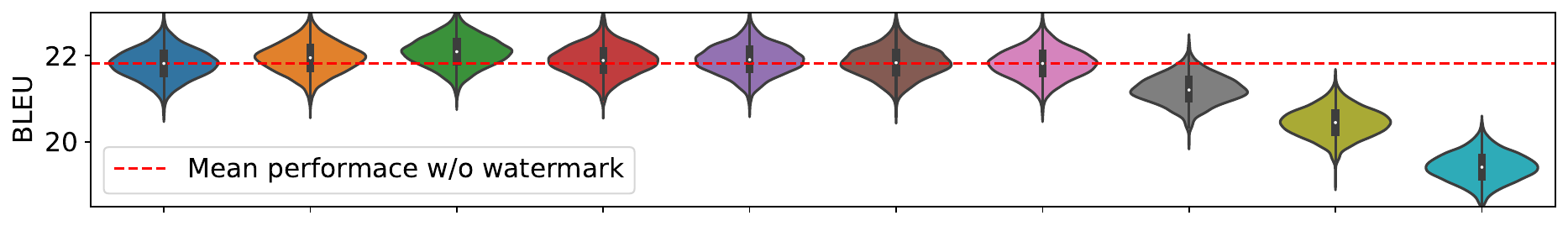}
    \includegraphics[width=1\textwidth]{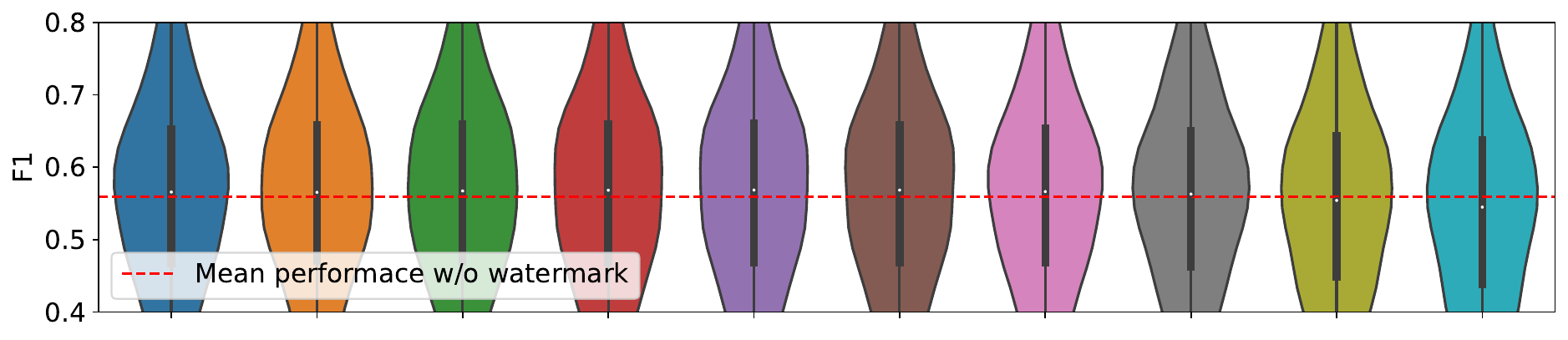}
    \includegraphics[width=1\textwidth]{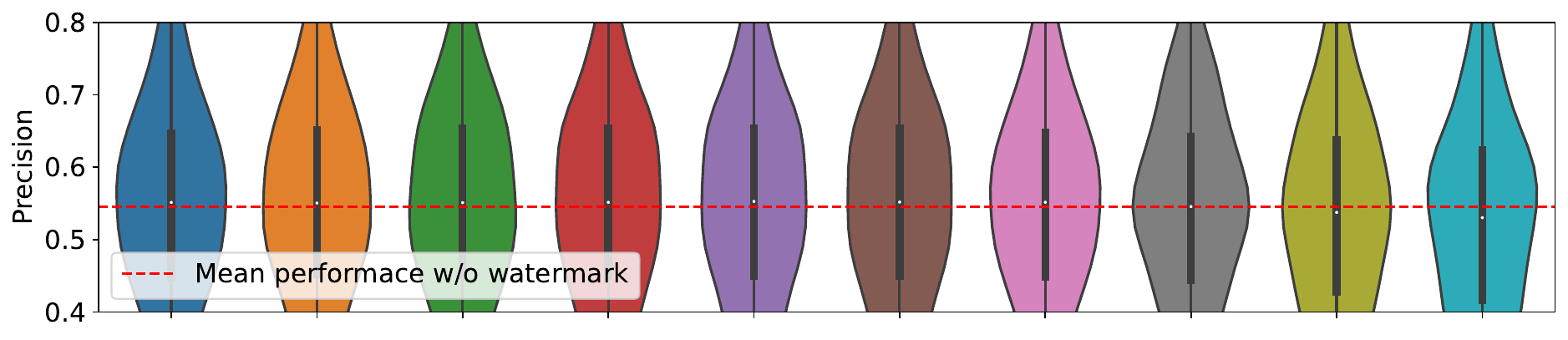}
    \includegraphics[width=1\textwidth]{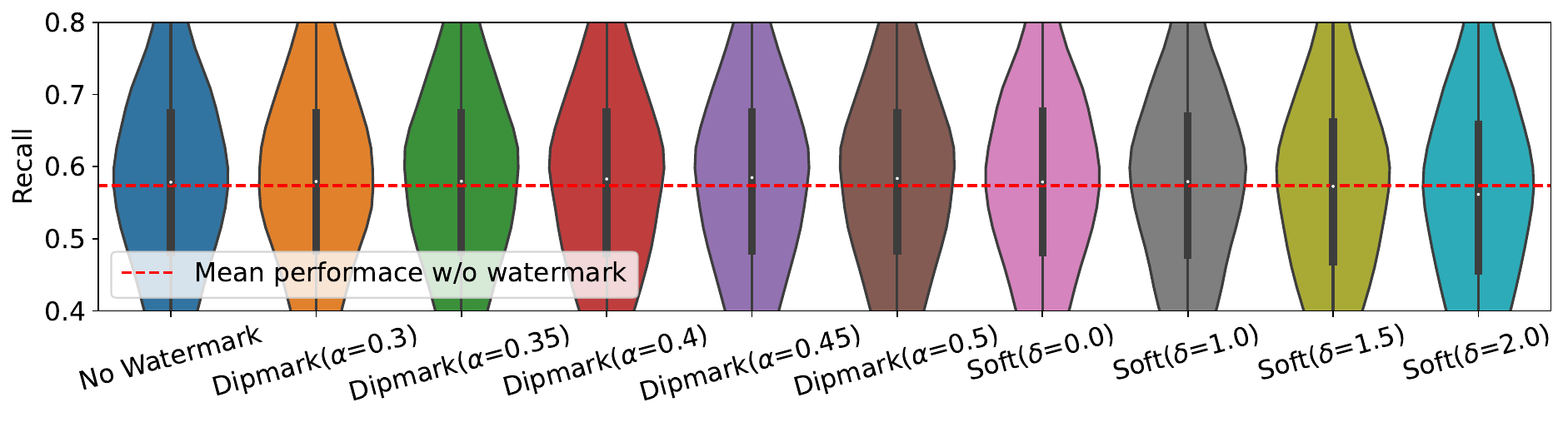}
    \caption{Violin plot of Machine Translation performance .}
    \label{fig:Perplexity MT add}
\end{figure}

\begin{figure}
    \centering
    \includegraphics[width=0.85\textwidth]{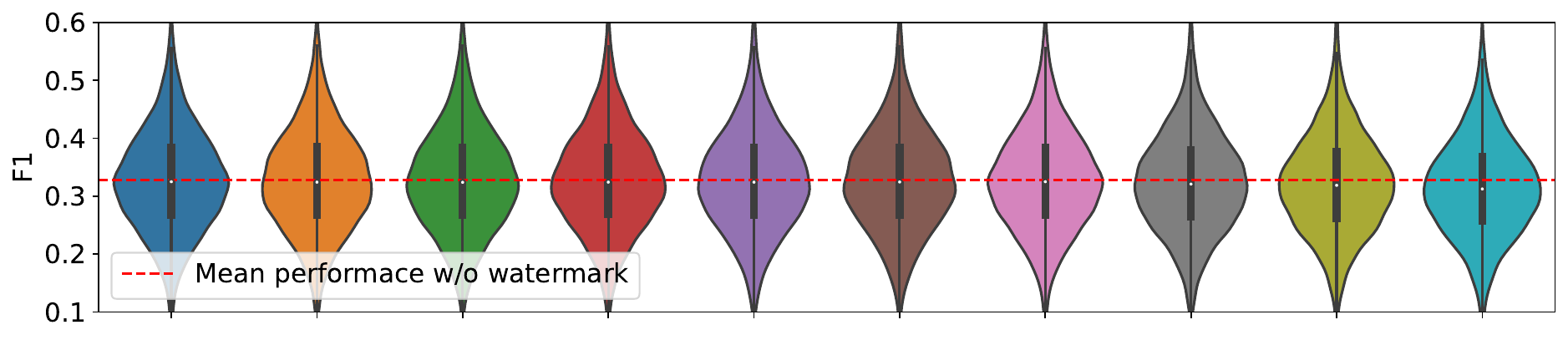}
    \includegraphics[width=0.85\textwidth]{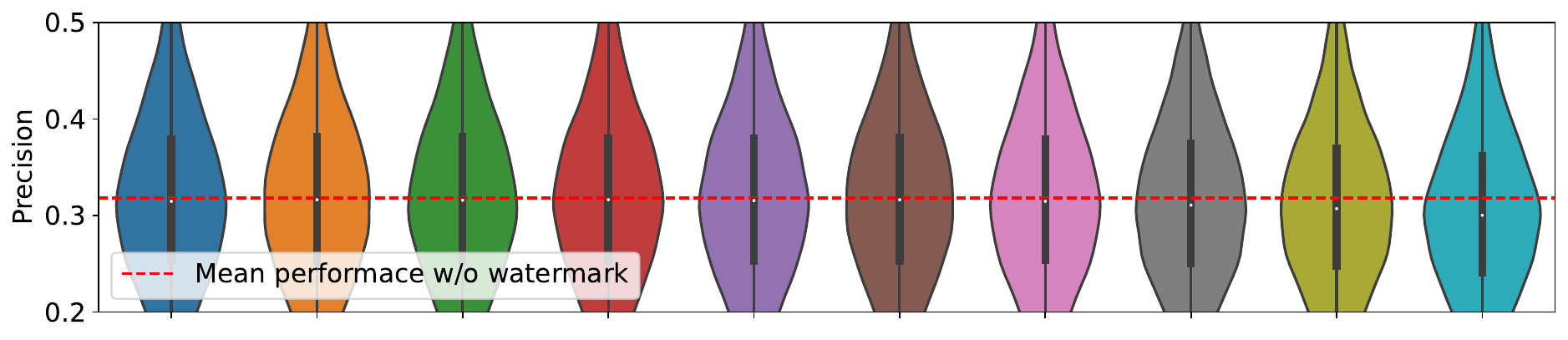}
    \includegraphics[width=0.85\textwidth]{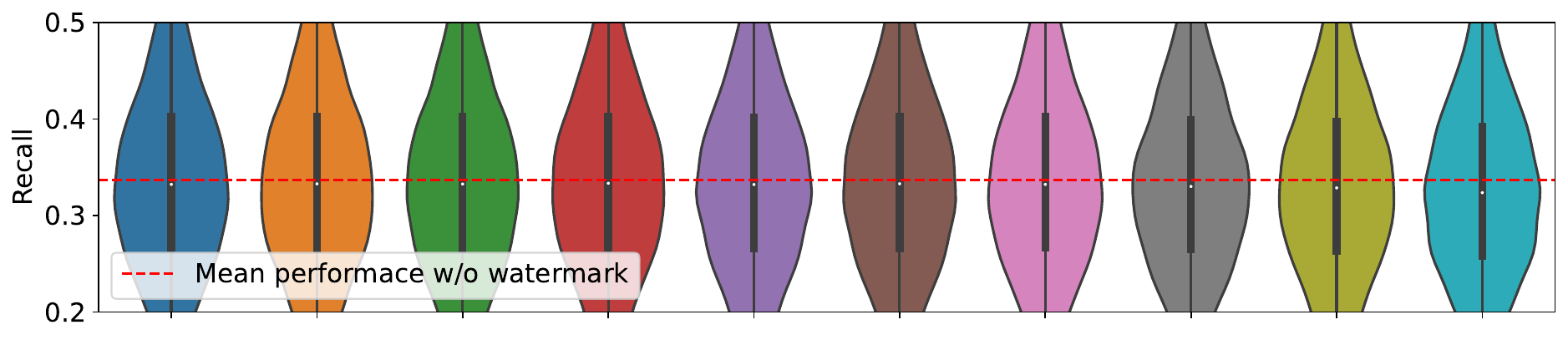}
    \includegraphics[width=0.85\textwidth]{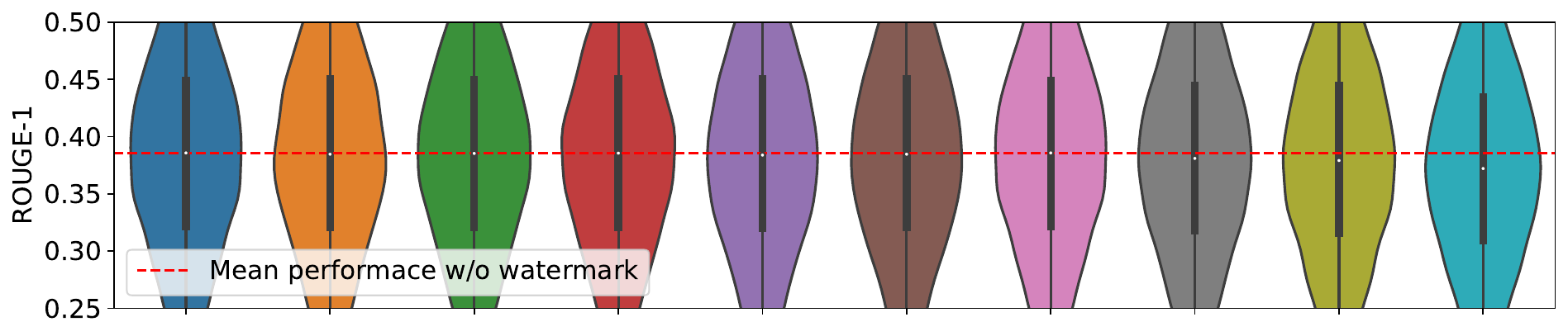}
    \includegraphics[width=0.85\textwidth]{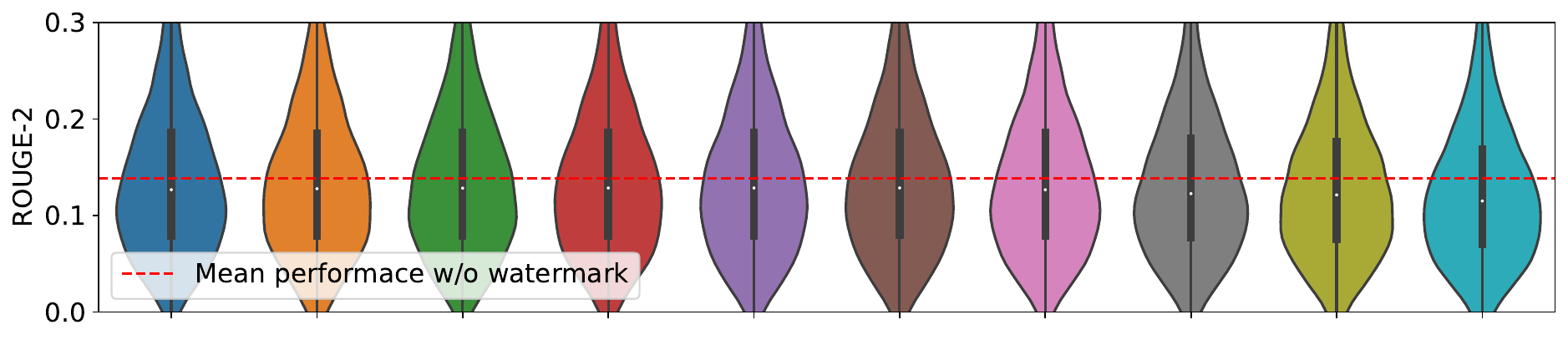}
    \includegraphics[width=0.85\textwidth]{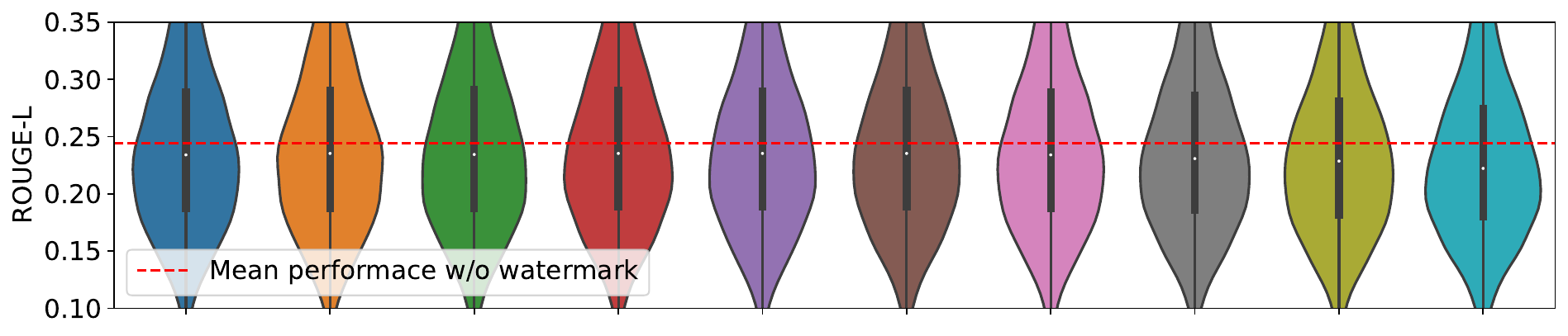}
    \includegraphics[width=0.85\textwidth]{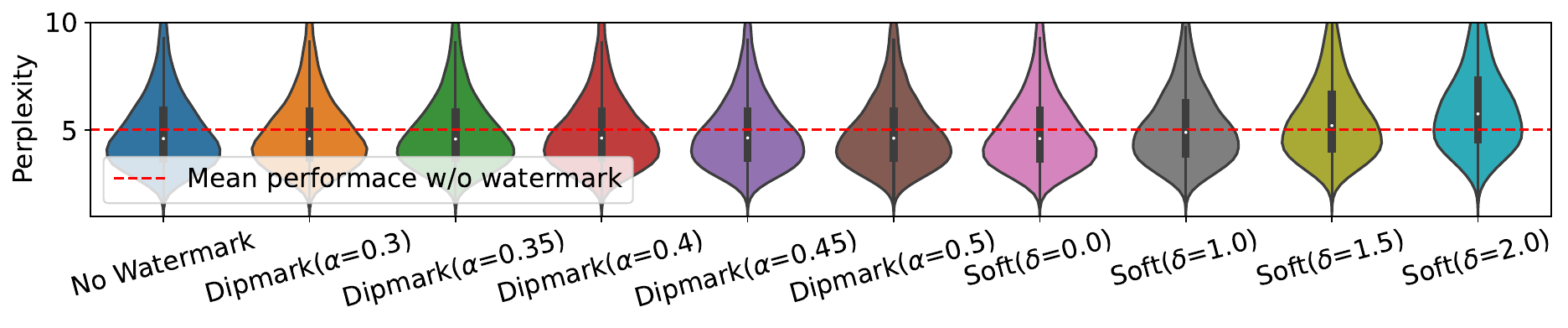}
    \caption{Violin plot of Text Summarization performance.}
    \label{fig:Perplexity TS add}
\end{figure}

\subsection{Distribution-preserving}\label{sec:add1}
\textbf{Settings.} 
In our evaluation, we assess the distribution-preserving performance of DiPmark within the context of two significant applications involving seq2seq models: machine translation (MT) and text summarization (TS). We follow the settings in \cite{hu2023unbiased}. For the TS task, our experimentation employs the BART-large model \citep{liu2020multilingual} in conjunction with the CNN-DM corpus \citep{hermann2015teaching} as our designated testing dataset. The MT task, on the other hand, revolves around English-to-Romanian translation. For this purpose, we employ the Multilingual BART (MBart) model \citep{liu2020multilingual} on the WMT’14 En-Ro corpus. Specifically for DiPmark, we select values for $\alpha$ from the set $\{0.3, 0.35, 0.4, 0.45, 0.5\}$, while for the Soft watermark \citep{kirchenbauer2023watermark}, we choose green list bias values $\delta$ from the set $\{0.0, 1.0, 1.5, 2.0\}$ alongside a fixed green list separator $\gamma=0.5$, indicating that 50\% of tokens are green while the remainder are red. It is important to note that the Soft watermark with $\delta=0.0$ is essentially equivalent to no watermark since it does not promote the probability of green list tokens.

A thorough examination of Figure~\ref{fig:Perplexity MT add}, Figure~\ref{fig:Perplexity TS add}, Table~\ref{tab:machineadd}, and Table~\ref{tab:summaradd} reveals a discernible trend. Throughout the range of $\alpha$ values spanning $\{0.3, 0.35, 0.4, 0.45, 0.5\}$, all the metrics associated with machine translation tasks and text summarization tasks maintain a consistent alignment between DiPmark and the original language model. Conversely, an upward adjustment in the $\delta$ values of the Soft watermark distinctly impacts the quality of the text output.
\begin{figure}
    \centering
    \includegraphics[width=0.495\textwidth]{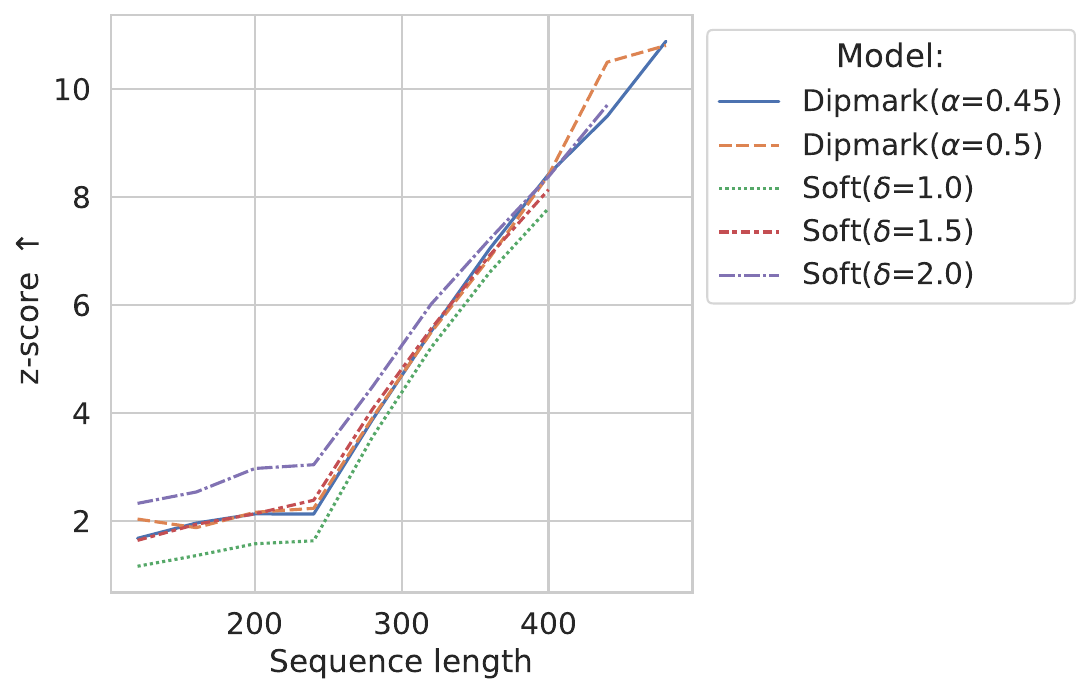}
    \includegraphics[width=0.495\textwidth]{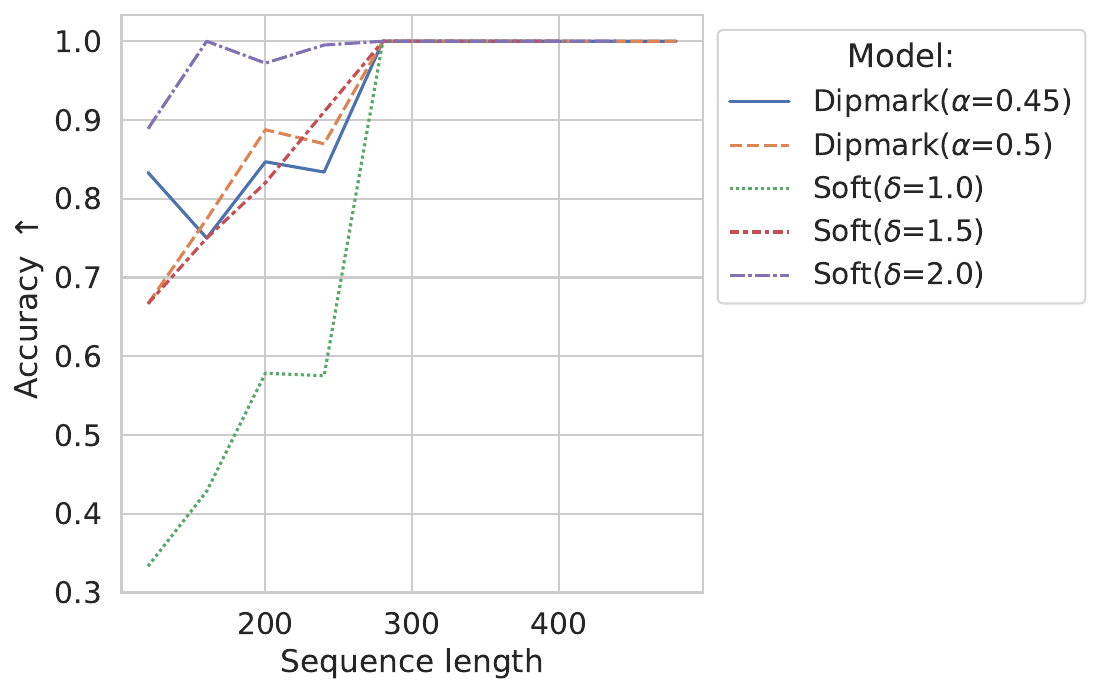}
    \caption{\textbf{Left:} Average z-score vs token sequence length with $\gamma=0.5$ on text generation tasks. \textbf{Right:} Watermark detection accuracy vs token sequence length with $\gamma=0.5$ and threshold $z = 1.517$ (false positive rate less than 0.01) on text generation tasks.}
    \label{fig:det2}
\end{figure}

\begin{figure}[t]
    \centering
    \includegraphics[width=0.495\textwidth]{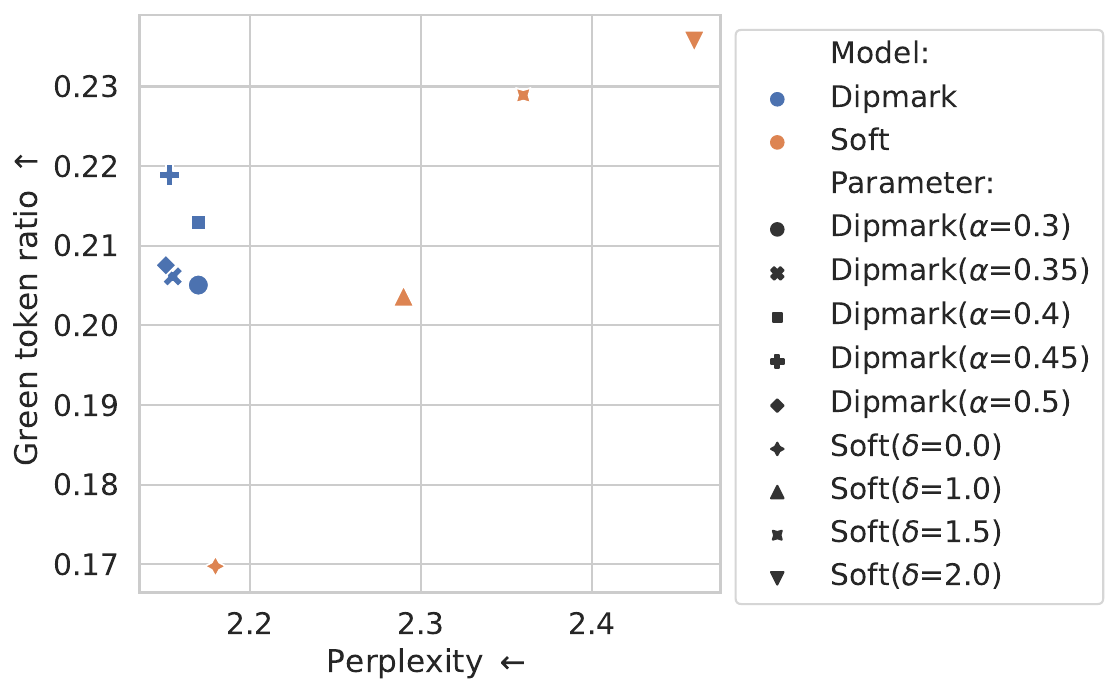}
    \includegraphics[width=0.495\textwidth]{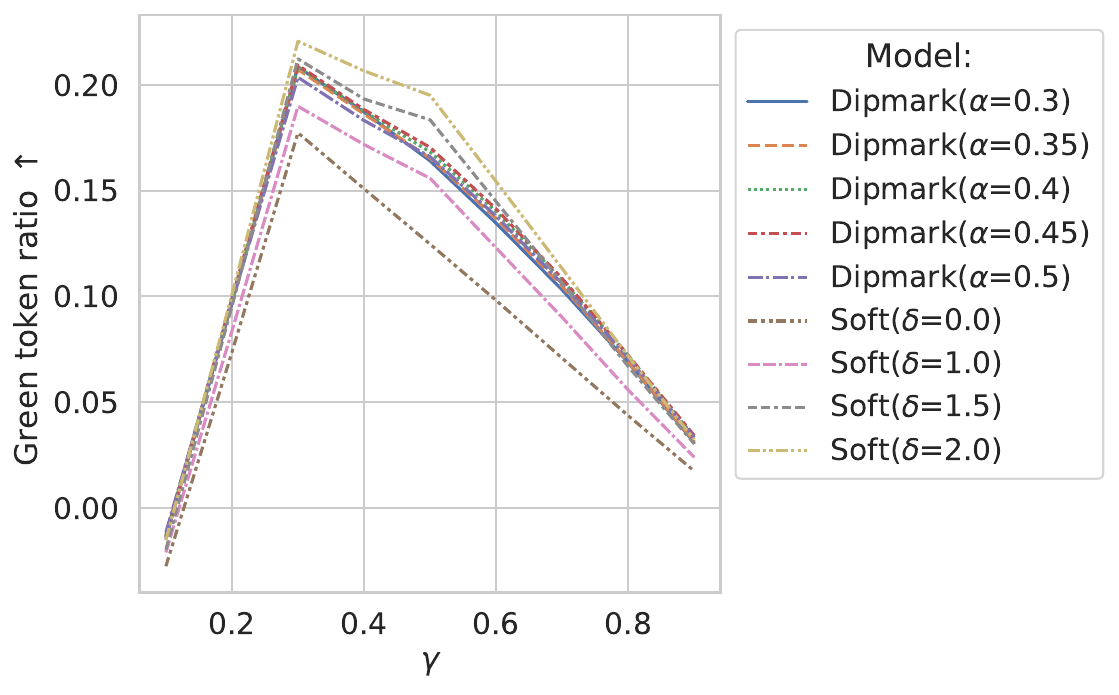}
    \caption{\textbf{Left.} Average Perplexity vs Green token rate with $\gamma=0.5$ on the text summarization task. \textbf{Right.} Avg. Green token ratio with different $\gamma$ on the text summarization task.}
    \label{fig:det1ad}
\end{figure}
\begin{figure}[t]
    \centering
    
    \includegraphics[width=0.495\textwidth]{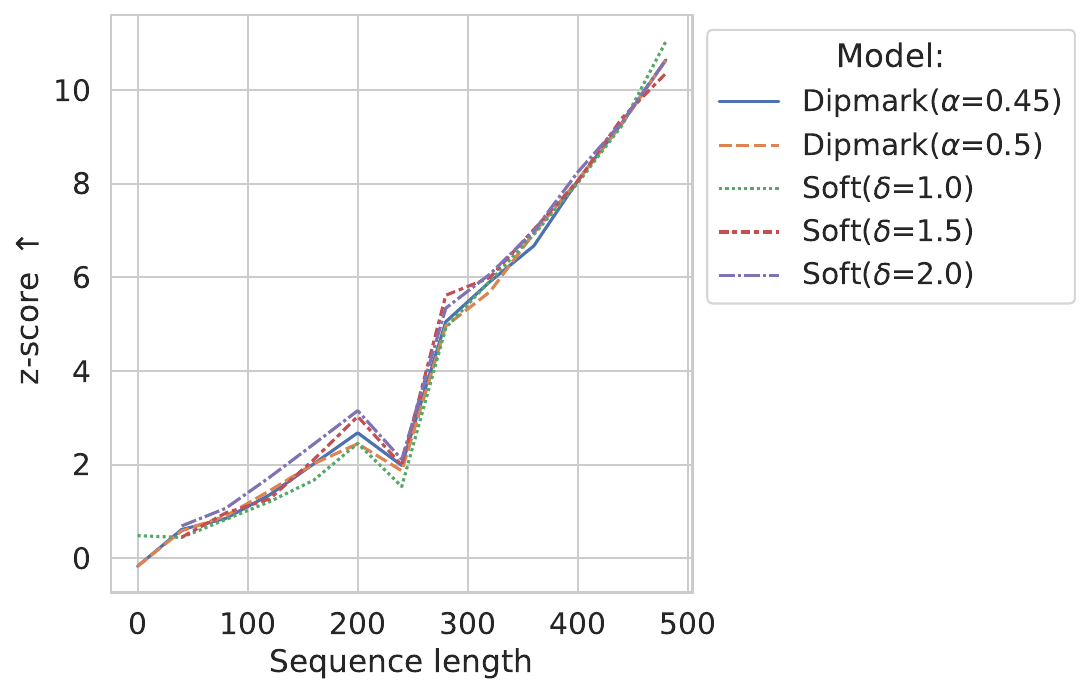}
    \includegraphics[width=0.495\textwidth]{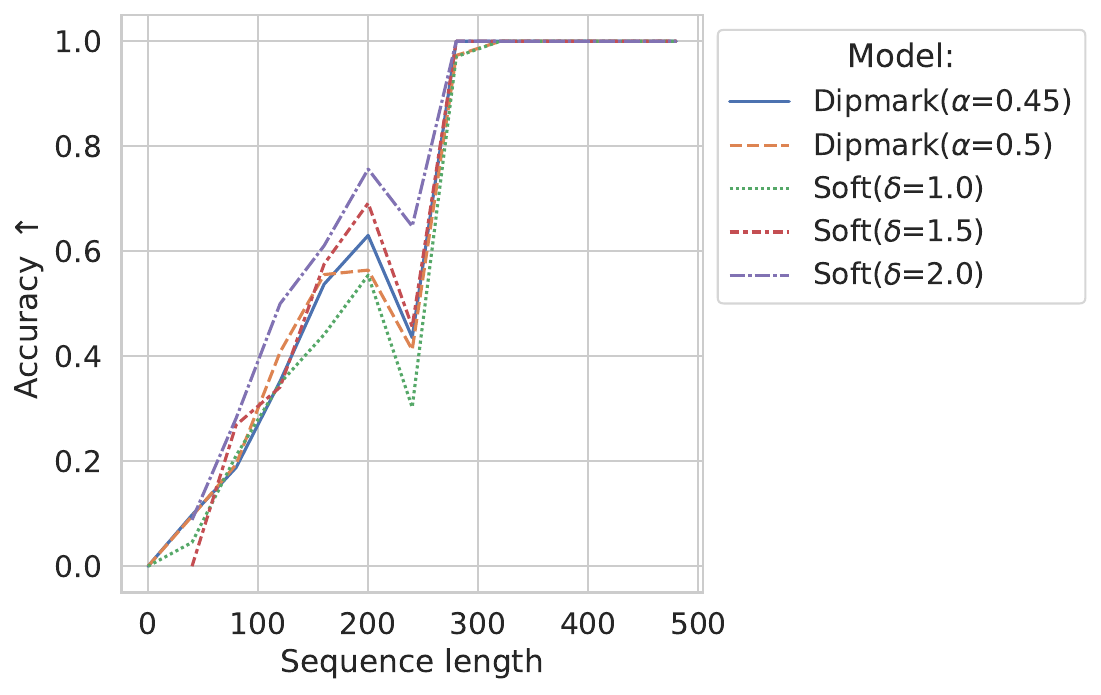}
    \caption{\textbf{Left.} Average z-score vs token sequence length with $\gamma=0.5$ on the text summarization task. \textbf{Right.}Avg. best p-score with text length with $\gamma=0.5$ on the text summarization task.  }
    \label{fig:det2ad}
\end{figure}
\subsection{Detectability comparison}\label{sec:add2}

\textbf{Settings.} 
We evaluate the detectability of our watermark on text summarization tasks using LLaMA-2. We generate 1,000 examples for each tasks.
 We also select $\alpha \in\{ 0.3, 0.35, 0.4, 0.45, 0.5\}$ for DiPmark, and $\delta \in \{0.0, 1.0, 1.5, 2.0\}$ and $\gamma=0.5$ for Soft watermark \citep{kirchenbauer2023watermark}. During detection, we also use $\gamma=0.5$. We report the green token ratio (defined in \ref{sec:detection}), the score of $\Phi(\gamma,\x)$ (z-score), and the detect accuracy.

\textbf{Result analysis.}  The results for text generation are visually depicted in Figure~\ref{fig:det1} and Figure~\ref{fig:det2}. Broadly speaking, our DiPmark variants with $\alpha = 0.45$ and $0.5$ exhibit performance comparable to that of the Soft watermark with $\delta = 1.5$, where $\delta=1.5$ corresponds to an augmentation of 1.5 to the green token logits. In Figure~\ref{fig:det1} (left), it is evident that our DiPmark variants with $\alpha=0.45$ and $0.5$ yield green token ratios akin to those of the Soft watermark with $\delta=1.5$ without any discernible degradation in text quality. Figure~\ref{fig:det1} (right) delves into the impact of different green list separators $\gamma$, revealing that, for most watermark models, $\gamma=0.5$ yields the highest green token ratio, underscoring its suitability as a reasonable choice for watermark detection. In Figure~\ref{fig:det2} (left) and Figure~\ref{fig:det2} (right), we present the average z-scores and accuracy metrics relative to sequence length. It is conspicuously observable that longer token sequences tend to facilitate easier detection, in line with our earlier analysis in Section~\ref{sec:detection}.
 The results for text summarization are visually depicted in Figure~\ref{fig:det1ad} and Figure~\ref{fig:det2ad}. Broadly speaking, our DiPmark variants with $\alpha = 0.45$ and $0.5$ exhibit performance comparable to that of the Soft watermark with $\delta = 1.5$, where $\delta=1.5$ corresponds to an augmentation of 1.5 to the green token logits. In Figure~\ref{fig:det1ad} (left), it is evident that our DiPmark variants with $\alpha=0.45$ and $0.5$ yield green token ratios akin to those of the Soft watermark with $\delta=1.5$ without any discernible degradation in text quality. Figure~\ref{fig:det1ad} (right) delves into the impact of different green list separators $\gamma$. Interestingly, for most watermark models, $\gamma=0.3$ yields the highest green token ratio instead of $\gamma=0.5$, which may be due to the low entropy characteristic of the text summarization task. In Figure~\ref{fig:det2ad} (left) and Figure~\ref{fig:det2ad} (right), we present the average z-scores and accuracy metrics relative to sequence length. It is conspicuously observable that longer token sequences tend to facilitate easier detection, in line with our earlier analysis in Section~\ref{sec:detection}. 
 
\subsection{Resilience}\label{sec:add3}

 We conduct experiments to test the resiliency of the our DiPmark and the Soft watermark in \cite{kirchenbauer2023watermark}. In this context, we use the text summarization tasks with 1,000 generated sequences on LLaMA-2. For resilience evaluation, we manipulating about $\epsilon\in\{0.05, 0.1, 0.2, 0.3\}$ portion of the text tokens through text insertion, text substitution, and text deletion.

 \textbf{Result Analysis.} 
 Figure~\ref{fig:addrobust} elucidates the evolution of the average green token ratio and the average z-score concerning the attack strength parameter $\epsilon$. Notably, both metrics exhibit a diminishing trend as $\epsilon$ increases.

\begin{figure}
    \centering
    \includegraphics[width=0.495\textwidth]{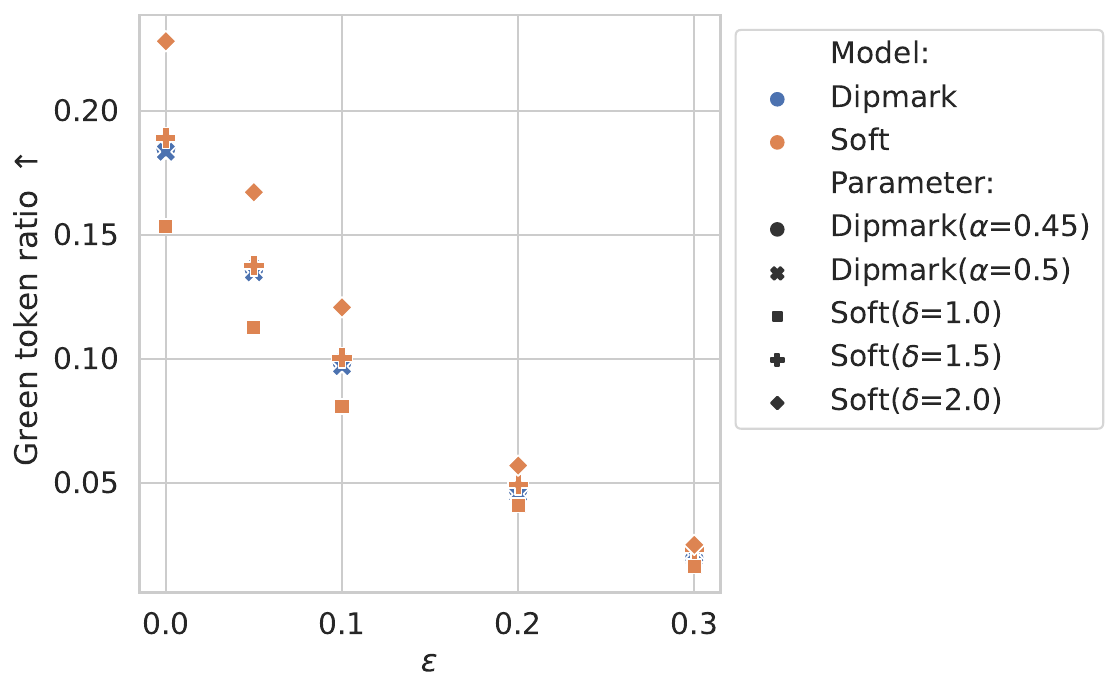}
    \includegraphics[width=0.495\textwidth]{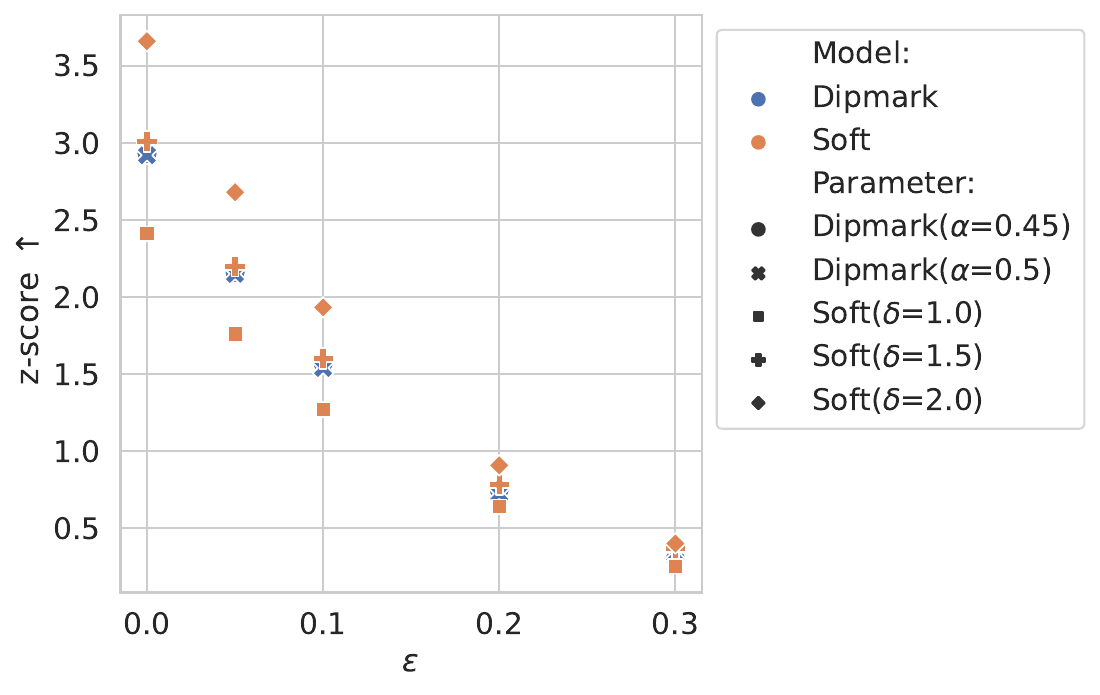}
    \vspace{-0.6cm}
    \caption{Robustness evaluation of DiPmark on text generation task. \textbf{Left.} Average green token ratio w.r.t. portion of perturbation $\epsilon$. \textbf{Right.} Average z-score w.r.t. portion of perturbation $\epsilon$.}
    \label{fig:robust1}
    \vspace{-0.2cm}
\end{figure}
 \begin{figure}
    \centering
    \includegraphics[width=0.495\textwidth]{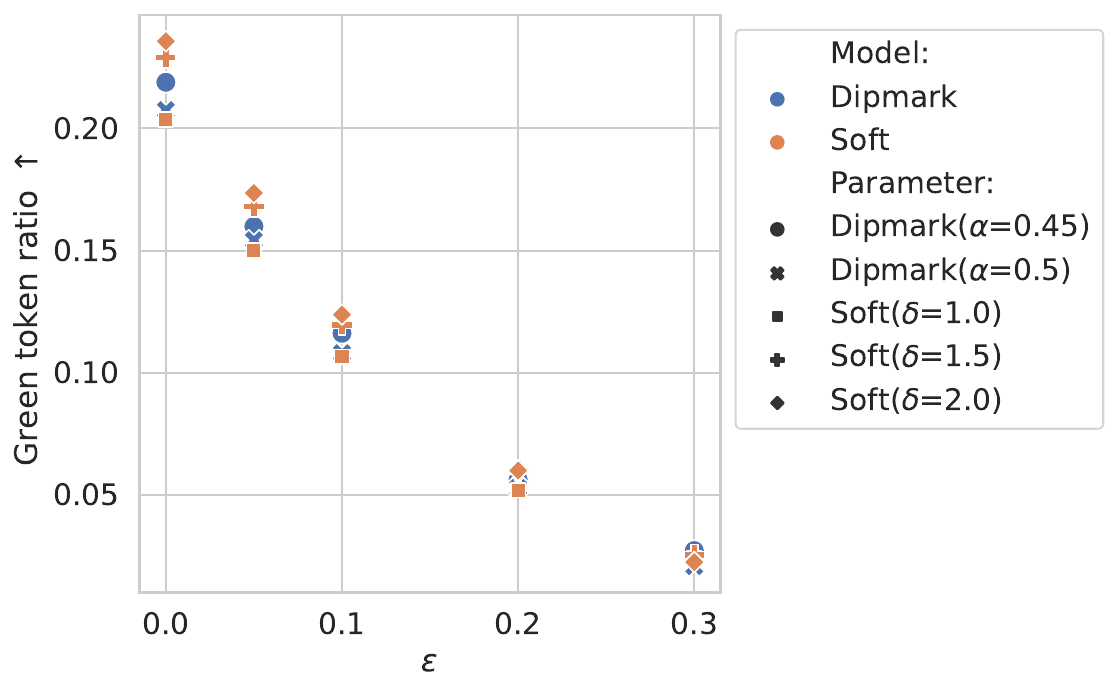}
    \includegraphics[width=0.495\textwidth]{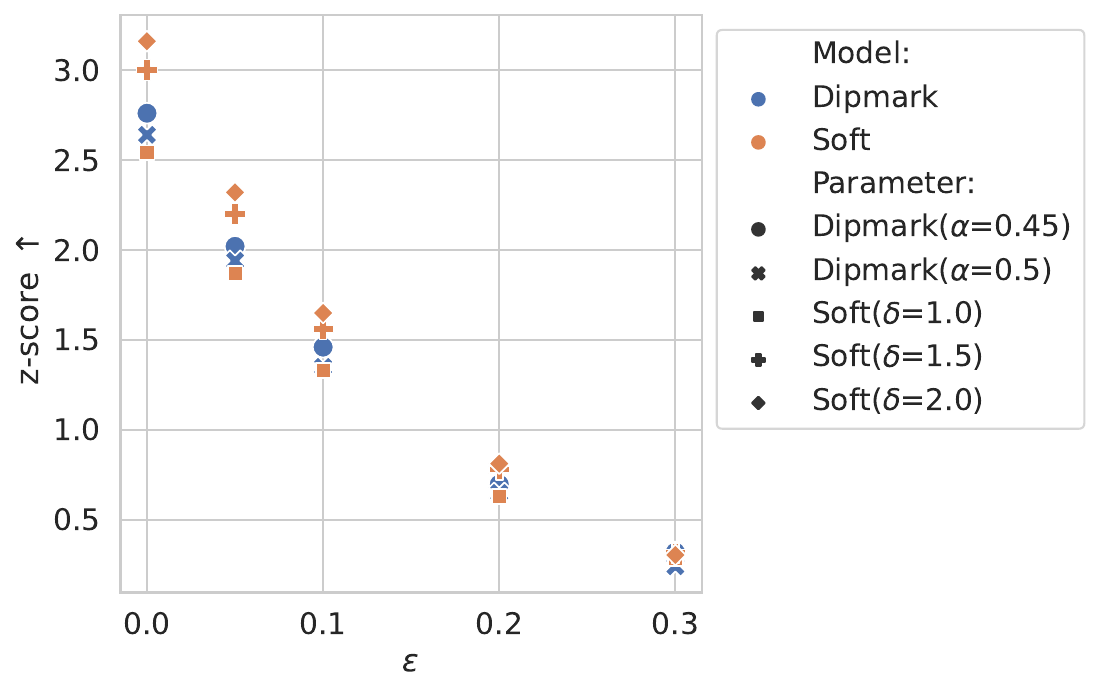}
    \caption{Robustness evaluation of DiPmark on text summarization task. \textbf{Left.} Average green token ratio w.r.t. portion of perturbation $\epsilon$. \textbf{Right.} Average z-score w.r.t. portion of perturbation $\epsilon$.}
    \label{fig:addrobust}
\end{figure}

\section{Broader Impacts}\label{sec:broader impacts}

Machine learning models exert substantial influence across various sectors, showcasing their capability to both improve efficiency and solve complex problems~\citep{yang2020prioritizing,yang2019lasso,wen2023feature,chakraborty2022using,cai2022asset,chen2024your,xu2017low,feng2018indexing}. Despite these benefits, concerns regarding the integrity and security of machine learning implementations persist~\citep{wu2023adversarial,wu2022retrievalguard,wu2023law,hong2024improving,pmlr-v202-hu23g,wang2023defending,wang2023distributionally}. In this setting, watermarking plays a crucial role by verifying the authenticity and ownership of digital media and aiding in the identification of AI-generated content.

\section{Examples of the watermarked text}
We list several examples of the watermarked text generated by LLaMA-2 on the text summarization task. We also report the p-value of the statistal testing using $\Phi(\gamma,\x_{1:n})$.
\begin{figure}[t]
    \centering
    \includegraphics[width=0.8\textwidth]{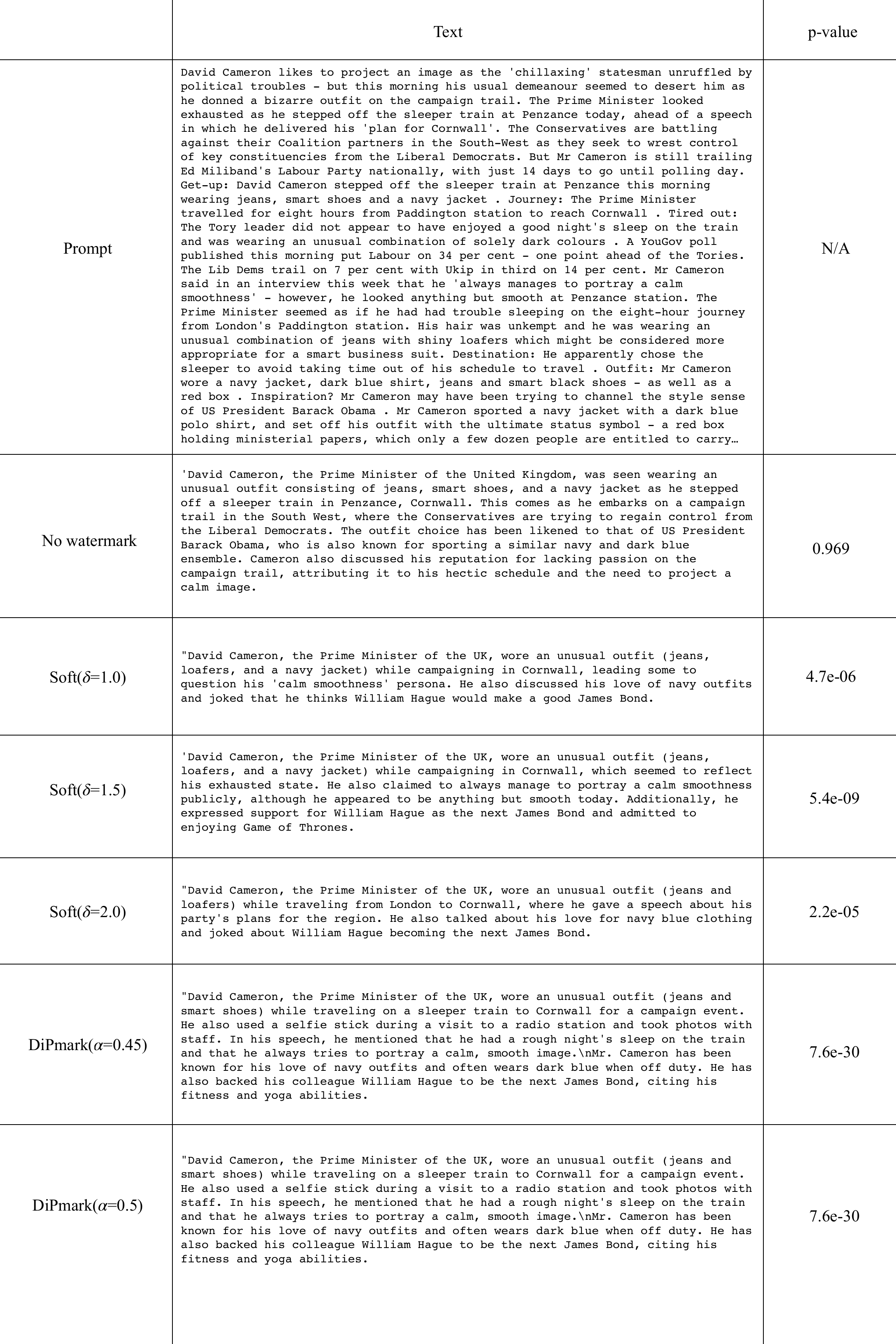}
    \caption{Examples of the watermarked text generated by LLaMA-2 on text summarization tasks.}
    \label{fig:text1}
\end{figure}

\begin{figure}[t]
    \centering
    \includegraphics[width=0.8\textwidth]{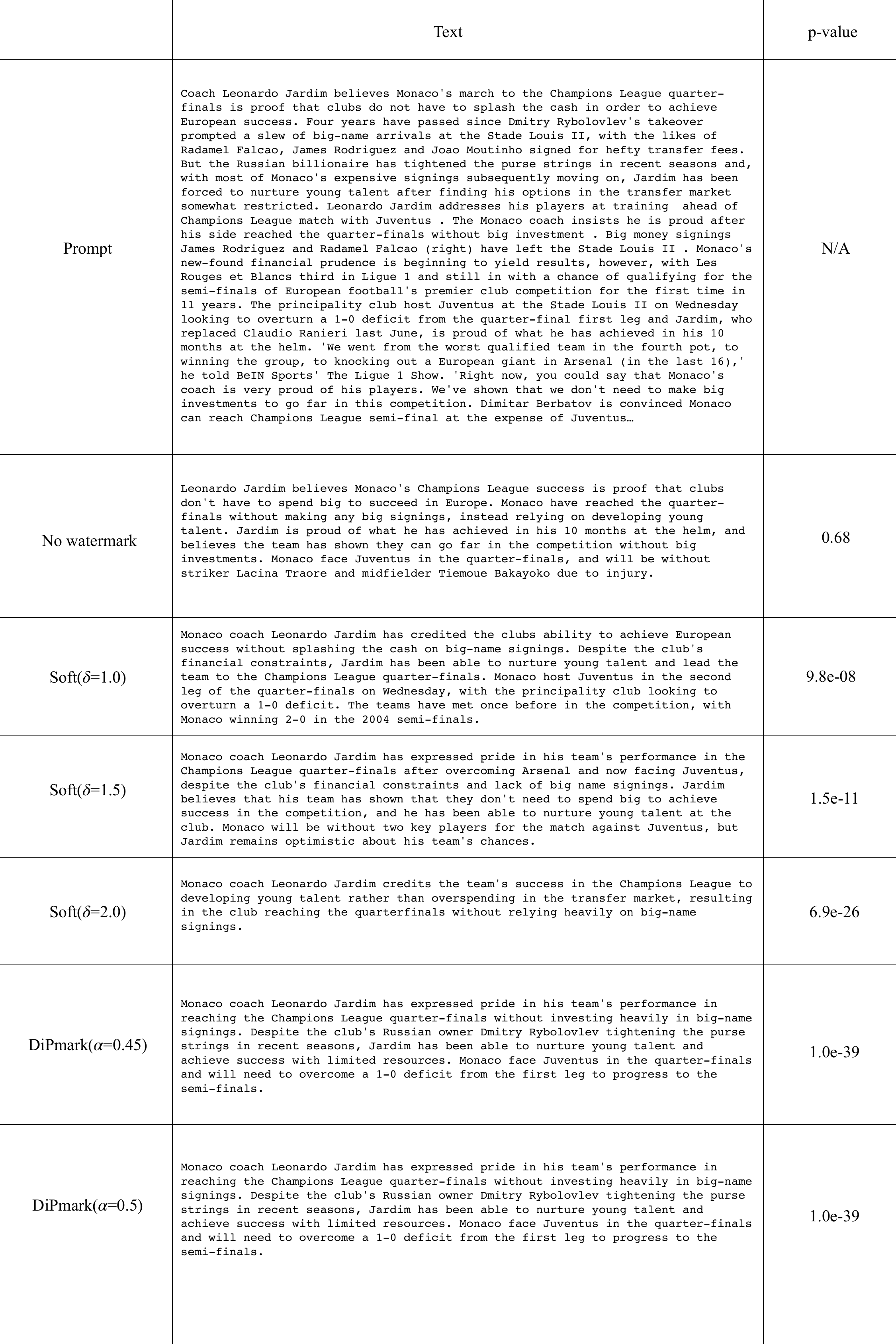}
    \caption{Examples of the watermarked text generated by LLaMA-2 on text summarization tasks.}
    \label{fig:text2}
\end{figure}

\begin{figure}[t]
    \centering
    \includegraphics[width=0.8\textwidth]{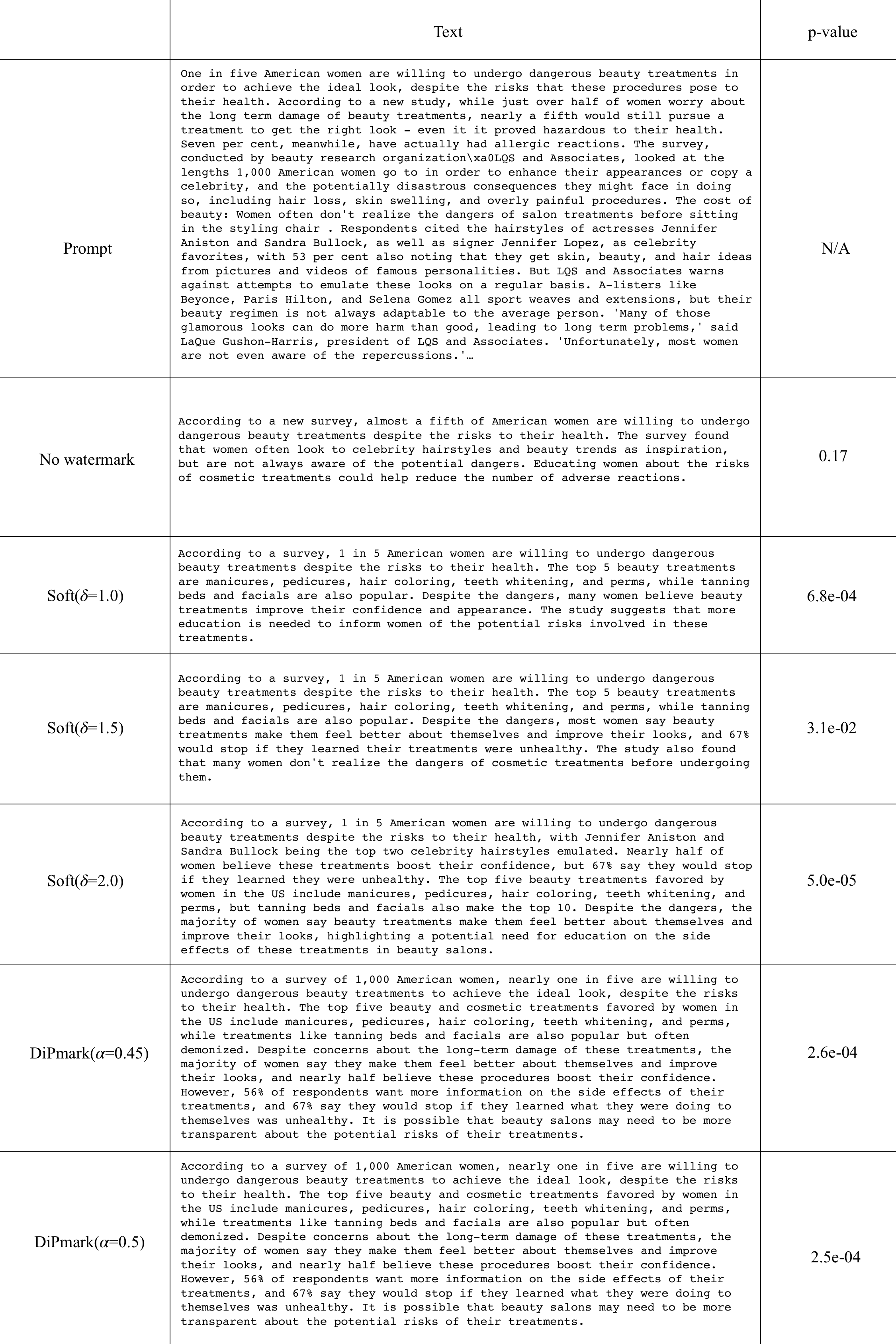}
    \caption{Examples of the watermarked text generated by LLaMA-2 on text summarization tasks.}
    \label{fig:text3}
\end{figure}

\begin{figure}[t]
    \centering
    \includegraphics[width=0.8\textwidth]{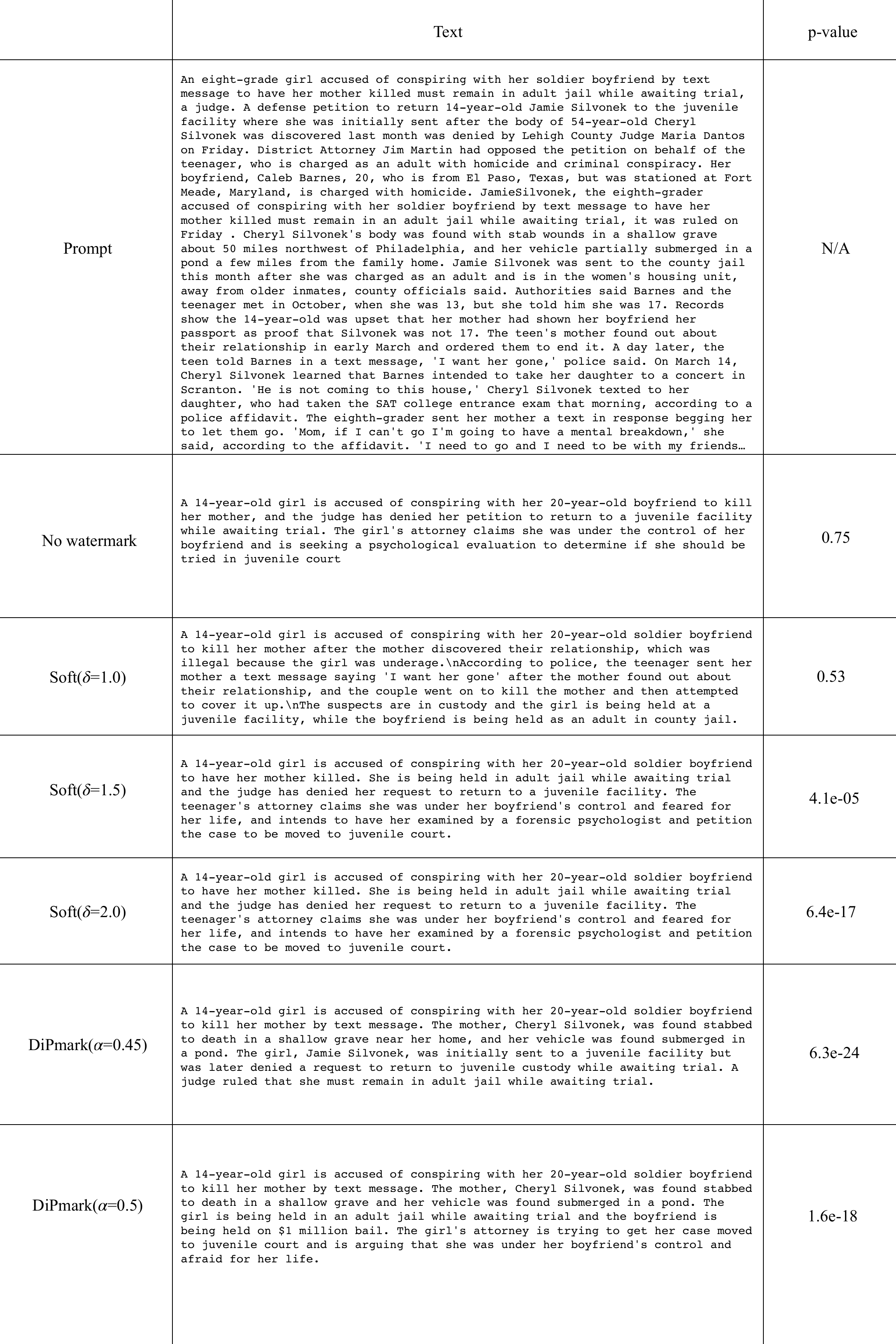}
    \caption{Examples of the watermarked text generated by LLaMA-2 on text summarization tasks.}
    \label{fig:text4}
\end{figure}

\begin{figure}[t]
    \centering
    \includegraphics[width=0.8\textwidth]{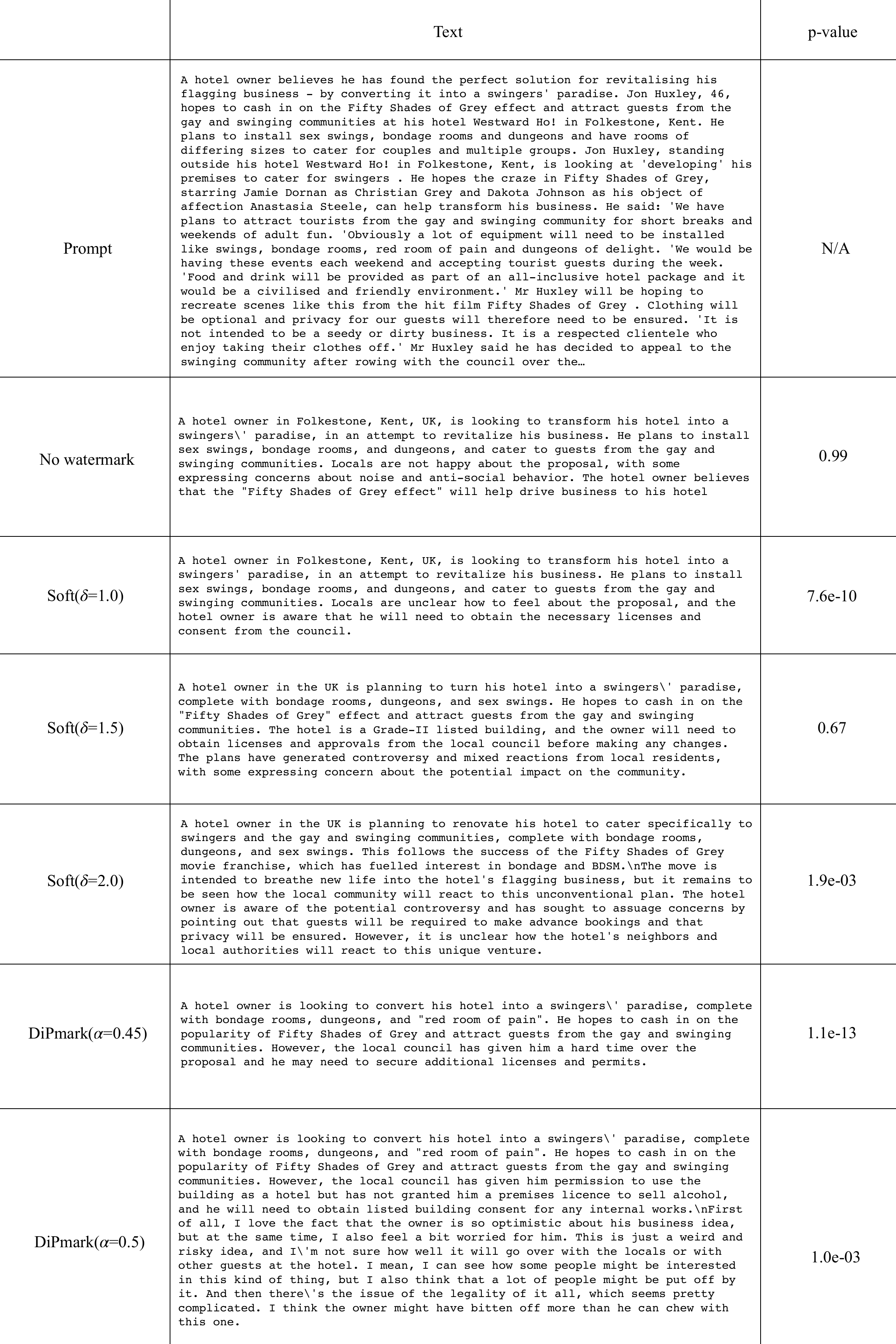}
    \caption{Examples of the watermarked text generated by LLaMA-2 on text summarization tasks.}
    \label{fig:text5}
\end{figure}

\end{document}